%% file: arxiv.tex
\documentclass{article}
\usepackage{arxiv}

\usepackage{tikz}
\usetikzlibrary{shapes,backgrounds,patterns,calc}
\usepackage{booktabs} %
\usepackage{hyperref}
\usepackage{macros}

\usepackage{algorithm}
\usepackage{algorithmic}

\input{pkgs}

\usepackage{natbib}
\usepackage[capitalize,noabbrev]{cleveref}
\input{defs}

\title{Agent-Designed Contracts: \\
	How to Sell Hidden Actions}

\author{
	Martino Bernasconi\\
	Bocconi University\\
	\texttt{martino.bernasconi@unibocconi.it}
	\And
	Matteo Castiglioni\\
	Politecnico di Milano\\
	\texttt{matteo.castiglioni@polimi.it}
	 \And
	Andrea Celli\\
	Bocconi University\\
	\texttt{andrea.celli2@unibocconi.it}
}

\begin{document}
	
\maketitle

\begin{abstract}
	We study the problem faced by a \emph{service provider} that has to sell services to a \emph{user}. In our model the service provider proposes various payment options (a menu) to the user which may be based, for example, on the quality of the service. Then, the user chooses one of these options and pays an amount to the service provider, contingent on the observed final outcome. 
	Users are not able to observe directly the action performed by the service provide to reach the final outcome. This might incentivize misconduct. Therefore, we propose a model that enforces trust through economics incentives.
	The problem has two crucial features: i) the service provider is responsible for \emph{both} formulating the contract and performing the action for which the user issues payments, and ii) the user is unaware of the true action carried out by the service provider, which is \emph{hidden}.
	We study this \emph{delegation problem} through the lens of contract design, with the overarching goal of enabling the computation of contracts that guarantee that the user can \emph{trust} the service provider, even if their action is hidden. 
	
	We study the Bayesian problem in which users are characterized by $n$ types. We start by showing that there is no polynomial algorithm that can approximate the optimal menu within any additive factor, \ie the problem is not in $\mathsf{APX}$. However, we find that the problem can be efficiently solved if the service provider is forced to provide a menu of constant size. 
	Our algorithm exploits a reduction to a related multi-item pricing problem with unit-demand and price-floors, and take advantage of the reduced dimensionality of this problem. 
	Then, we extend the computational problem in two directions. First, we study the case of continuous actions, which is crucial for applications. Surprisingly, the technical analysis of this setting also enables the computation of menus that are robust with respect to errors in the description of the input instance. Second, we extend the model to handle menus of randomized payments, and we show that optimal menus of randomized payments can be computed in polynomial time and provide at least $\Omega(n)$ times more utility than optimal deterministic menus.
\end{abstract}

\input{src/introduction}

\input{src/related}

\input{src/prelim}

\input{src/NPHard}
\input{src/single}

\input{src/single2}

\input{src/powerOfClasses}

\input{src/monotonenew}
\input{src/random}

\section*{Acknowledgements}
Matteo Castiglioni is supported by the FAIR (Future Artificial Intelligence Research) project, funded by the NextGenerationEU program within the PNRR-PE-AI scheme (M4C2, Investment 1.3, Line on Artificial Intelligence), and by the EU Horizon project ELIAS (European Lighthouse of AI for Sustainability, No. 101120237). Andrea Celli and Martino Bernasconi are supported by the European Union -  NextGenerationEU, in the framework of the FAIR - Future Artificial Intelligence Research project (FAIR PE00000013 – CUP B43C22000800006), and by MUR - PRIN 2022 project 2022R45NBB funded by the NextGenerationEU program.

\bibliographystyle{ACM-Reference-Format}
\bibliography{biblio}

\clearpage
\appendix
\input{src/appendix}
\input{src/appendixRandom}
	
\end{document}

%% file: pkgs.tex
\usepackage{amsmath}
\usepackage{pict2e,picture,graphicx}
\usepackage{microtype}
\usepackage{graphicx}
\usepackage{subfigure}
\usepackage{booktabs}
\usepackage{cases}
\usepackage{mathtools}
\usepackage{amsthm}
\usepackage{thm-restate}
\usepackage{enumitem}
\usepackage{xcolor}
\usepackage{nicefrac, xfrac}

\usepackage{xparse}
\usepackage{leftindex}
\usepackage{bbm}
\usepackage{yfonts}
\usepackage{stmaryrd}
\usepackage{mathtools}
\usepackage{xspace}
\usepackage{amsfonts}
\usepackage{amssymb}
\usepackage{float}
\usepackage{hyperref}

%% file: defs.tex
\newcounter{programCounter}
\crefname{programCounter}{Program}{Programs}

\makeatletter
\DeclareRobustCommand{\Arrow}[1][]{%
	\check@mathfonts
	\if\relax\detokenize{#1}\relax
	\settowidth{\dimen@}{$\m@th\rightarrow$}%
	\else
	\setlength{\dimen@}{#1}%
	\fi
	\sbox\z@{\usefont{U}{lasy}{m}{n}\symbol{41}}%
	\begin{picture}(\dimen@,\ht\z@)
		\roundcap
		\put(\dimexpr\dimen@-.7\wd\z@,0){\usebox\z@}
		\put(0,\fontdimen22\textfont2){\line(1,0){\dimen@}}
	\end{picture}%
}
\makeatother

\makeatletter
\newcommand{\pushright}[1]{\ifmeasuring@#1\else\omit\hfill$\displaystyle#1$\fi\ignorespaces}
\newcommand{\pushleft}[1]{\ifmeasuring@#1\else\omit$\displaystyle#1$\hfill\fi\ignorespaces}
\newcommand{\specialcell}[1]{\ifmeasuring@#1\else\omit$\displaystyle#1$\ignorespaces\fi}
\makeatother

\newtheorem{corollary}{Corollary}
\newtheorem{assumption}{Assumption}

\newtheorem{theorem}{Theorem}
\newtheorem{lemma}{Lemma}

\usepackage{pifont}
\usepackage[normalem]{ulem} %
\definecolor{mygreen}{rgb}{0.0, 0.5, 0.0}
\definecolor{myorange}{rgb}{1, 0.7, 0.1}

\newtheorem*{theorem-non}{Theorem}
\newtheorem{definition}[theorem]{Definition}
\newtheorem{claim}[theorem]{Claim}

%% file: src/introduction.tex
\section{Introduction}

Contract design is a cornerstone of modern economic theory \cite{holmstrom1979moral,grossman1992analysis}. This framework describes how a principal can incentivize agents to perform favorable actions with moral hazard concerns. The computational study of this model has recently garnered significant attention \cite{dutting2019simple,dutting2021complexity,castiglioni2021bayesian,guruganesh2021contracts}.
One practical application in which contract design proved to be particularly relevant is to describe problems broadly related to machine learning pipelines, such as data collection~\cite{ananthakrishnan2023delegating}, data selling~\cite{chen2022selling}, and delegated classification~\cite{saig2023delegated}.

As machine learning models deployed across different sectors of society become larger and more complex, their training procedures become increasingly demanding. In this scenario, a party with limited computational resources (\ie the \emph{user}) may outsource the training procedure to a ``powerful'' untrusted party (\ie the \emph{service provider}). 
In this setting, the service provider is usually in charge of proposing to the user some payment options. This is in striking contrast with the standard \emph{principal-agent} set-up, in which the principal proposes a contract and the agent performs an action \cite{holmstrom1979moral}.
To draw a parallelism between the two settings, in our terminology the service provider corresponds to the agent as it's the party performing the action, while the user is the principal, as it is the one delegating the task and lacks observability on the performed action. 
In this paper, we study a contract design problem which encapsulates the main features of this kind of delegation problems. In particular, the service provider (\ie the agent) is the designer of the contract, being the dominant party in the interaction.

In our model, the service provider is responsible for \emph{both} formulating the contract and performing the action for which the user issues payments. A user has a type drawn from a finite set according to a known distribution. Since the service provider is not fully aware of user's preferences, they usually propose a \emph{menu} of possible options to adapt to all different user preferences. As an example, modern cloud providers that offer machine learning training services allow users to select their preferred ‘‘machine'' and desired features.\footnote{For instance, see the pricing schemes of \href{https://aws.amazon.com/it/sagemaker/pricing/?p=pm&c=sm&z=2}{Amazon SageMaker} or \href{https://cloud.google.com/vertex-ai/pricing?hl=it}{Googles's Vertex AI}.} In general, a menu is a set of payment schemes, each of them binding outcomes to a desired payment. 
For instance, in our machine learning training example, outcomes may be defined in terms of the test error achieved by the trained model. 
In this context, oftentimes there is no method for the user to verify the actions undertaken by the service provider.
Therefore, the risk of moral hazard is usually ignored by assuming that the service provider behaves in a trustworthy way. 
In our model, we eliminate this assumption and establish trust through the design of suitable menus of payment schemes. These menus are carefully crafted so that the service provider has no interest in deviating from the user's selected action. 
In order to achieve that, our model resorts to outcome-based payment schemes, in which payments are contingent on the final outcomes reached through the service provider's actions.

\subsection{Challenges, Results, and Techniques}

We provide a computational characterization of the problem faced by a service provider who has to design a contract under Bayesian uncertainty on the user's type. 

\xhdr{Negative result.} 
We start by focusing on the problem of computing optimal menus of \emph{deterministic payment schemes}. A deterministic payment scheme maps each action of the service provider to a payment vector that specifies a payment for each possible outcome. An optimal menu may consist of multiple distinct payment schemes.
 In \Cref{thm:negative}, we show that, unless $\mathsf{P}=\mathsf{NP}$,  for any constant $\rho>0$ it is impossible to compute in polynomial time a multiplicative approximation up to within a factor $\rho$ of the service provider's optimal expected utility for an optimal menu of deterministic payment schemes.

\xhdr{Optimal Simple Menus.} Motivated by the negative results of \Cref{thm:negative}, we study the problem of computing \emph{menus with low complexity}, meaning that they include only a few payment schemes. 
This is also motivated by practical scenarios where service provider typically presents to users only a limited number of options.
The rationale for this approach is twofold: first, it simplifies and streamlines the interaction between the user and the service provider. Second, it has been empirically observed in diverse scenarios that users may be negatively influenced by behavioral biases leading to reduced purchases when presented with large sets of options \cite{thaler2015misbehaving,oulasvirta2009more,iyengar2000choice,sethi2004much,schwartz2004paradox}.
This model includes, as a special case, the problem of computing an optimal contract, that is the case in which the menu contains a single payment scheme. 
In standard principal-agent problems, computing optimal contracts is known to be \NPHARD \cite{guruganesh2021contracts,castiglioni2021bayesian}.
Surprisingly, in the case of agent-designed menus, an optimal menu can be computed in polynomial time whenever its size $k$ is constant. 
In standard principal-agent problem, the problem is made computationally intractable by the large outcome space, and it becomes tractable when the number of outcomes is constant  \cite{guruganesh2021contracts,castiglioni2021bayesian}.
We show how to solve the problem by reducing it to a \emph{unit-demand pricing problem with price-floors}, which has a reduced dimensionality with respect to the original problem. In this pricing problem, there is one item for sale for each of the payment options available in the original problem. The seller sets a price for each of these items. Then, a unit-demand buyer with a Bayesian type buys the item maximizing their utility. Finally, the item chosen by the buyer is produced by the seller, thereby incurring in some production costs. Interestingly, we observe that the incentive-compatibility constraints of the service provider translate into lower bounds on the allowed prices for each item, which we call ``price-floors''. We show that, given a solution to the pricing problem, one can compute in polynomial time a solution to the original contract-design problem, and vice-versa. Then, we show that if the number of items is constant, we can solve the pricing problem in polynomial time. This result has a number of immediate consequences: there exists a polynomial-time algorithm for the problem of computing an optimal agent-designed menu when the complexity of the menu is constant, when the number of types is constant, and when the number of actions of the service provider is constant. 

\xhdr{Trade-off between utility and complexity.} Following up on the positive results on the computation of optimal menus of constant size $k$, we investigate the efficiency of such simple menus. In particular, we study the fraction of the optimal service provider's utility which can be extracted by a deterministic menu of size $k$. 
We show that a menu of size $k$ can extract at least a $k/\min(n,\ell)$ fraction of the expected utility of an optimal menu, where $n$ is the number of user's types, and $\ell$ is the number of actions available to the service provider. 
To complete the picture, we show that this ratio is tight in terms of the worst-case maximum utility that can be extracted by using a menu of size $k$.

\xhdr{Robustness.} In practical situations, it is unlikely that the service provider has perfect knowledge of the precise reward function linked to each user's type. However, the service provider may be able to compute reasonable estimates of these rewards, for instance utilizing historical data. Then, a natural question is what can be achieved by the service provider if they only have access to estimates of the users' rewards. 
In this setting, the service provider can compute an optimal menu using an estimate of users' reward matrices. Therefore, such menu is IC and IR for the user only with respect to these estimates.
We make the assumption that outcome distributions and costs of similar actions should be sufficiently similar.
We show that, if the $\ell_\infty$-norm between the true and estimated reward matrices is at most $\delta$, then we can use the menu computed on the estimated rewards to build a new menu with the following features: i) the service provider's expected utility is at most an additive factor $O(\sqrt{\delta})$ away from the  utility of the original menu; ii) the new menu is almost IC for the service provider. 
The key observation is that both property i) and ii) hold even when users are taking decisions (\ie best-responding) on the basis of their true utilities, although they remain unknown to the service provider.  
We achieve this result through a carefully crafted ``reimbursement strategy'', according to which the service provider slightly decreases the magnitude of payments, and pays back the user a small fraction of their earnings.

\xhdr{Continuous actions.} The set of actions available to the service provider may be either continuous or so large to the extent that it may be regarded as continuous for practical purposes. These scenarios include our motivating application of delegating the training procedure of complex machine learning models. In this context, the action set may be the computing time allocated to a certain task. 
Here, our analysis makes the natural assumption that each outcome can be reached with sufficiently high probability, if an adequate effort is devoted to the task.
We show that this setting can be addressed with a two-phase algorithm. First, we solve a discretized version of the continuous problem under relaxed user’s IC and IR constraints. Then, we prove that this solution gives an expected utility to the service provider which is close to the value of an optimal solution of the problem with continuous actions. However, this solution only approximately satisfies IC and IR constraints for the user. Starting from here, an IC and IR menu for the user can be computed through the techniques we developed for computing robust menus. 

\xhdr{Randomization.} Finally, we explore whether exploiting randomization within menus can offer any advantage for the service provider. A menu of \emph{randomized payment schemes} is a set of payment schemes which, in this case, are composed of a probability distribution over actions and an action-dependent payment rule associating a payment to each outcome. When the user selects a randomized payment scheme, the action performed by the service provider is sampled from their ``randomized action-plan''.
We show that introducing randomization yields two major advantages for the service provider: i) first, there exist instances where the service provider's expected utility that can be extracted through a randomized menu is at least a multiplicative factor $\Omega(n)$ greater than the utility attained through an optimal deterministic menu; ii) second, an optimal menu of randomized payment schemes can be computed in polynomial time. Therefore, randomization allows us to circumvent the computational hardness result for deterministic menus. To achieve this, we show that the problem can be written as a quadratic optimization problem which can be relaxed through a linear program with polynomial size. This is in stark contrast with what happens in classical contract design problems, where an optimal menu of randomized payment schemes may not even exist in general, and one can only approximately reconstruct a solution from the relaxed problem \cite{castiglioni23design,gan2022optimal}.

%% file: src/related.tex
\subsection{Related Work}

The study of contract design from a computational perspective has received increasing attention in recent years.
Most of the research concentrates on single-agent contract design. In particular, \citet{babaioff2014contract} study the complexity of contracts in terms of the number of different payments that they specify, while \citet{dutting2019simple} analyze the efficiency of linear contracts with respect to optimal ones.
A recent line of works extends the analysis to Bayesian settings, in which the principal knows a probability distribution over agent's type~\cite{guruganesh2021contracts,castiglioni2021bayesian,alon2021contracts}.
\citet{castiglioni23design} introduce the class of menus of randomized contracts and show that they can be computed efficiently. \citet{gan2022optimal} study a generalization of the principal-agent problem to settings where the principal's action space may be infinite and subject to design constraints, and they prove that optimal randomized mechanisms can be computed efficiently.
Another line of research explored how to extend the single-agent model to scenarios with combinatorial actions \citep{dutting2021complexity,dutting2022combinatorial, dutting2024combinatorial, deo2024supermodular}. Finally, some works focus on settings with multiple agents~\cite{babaioff2006combinatorial,babaioff2012combinatorial,dutting2023multi,castiglioni23multi}.
Moreover, there is a line of works that analyzes robustness in classical contract design. In particular, \citet{carroll2015robustness, carroll2019robustness} studies the setting in which the principal only knows a small set of the actions (also called technologies in these works) available to the agent, and finds that linear contracts are worst-case optimal in this setting. \citet{dutting2019simple} studies a similar problem in which the principal has first-order ``moment information'' about the rewards of the actions, and performs a worst-case analysis in such setting. In the theoretical computer science literature, robustness was also considered in Stackelberg games \citep{gan2023robust} and Bayesian Persuasion \citep{babichenko2022regret,feng2024rationality}.

Similar to this paper, other works find inspiration in problems related with machine learning.
\citet{chen2022selling} consider the problem of selling data to a machine learner, who purchases data to train a machine learning model. They assume that the seller is the mechanism designer, and derive optimal and approximated selling mechanism for various settings.
\citet{ananthakrishnan2023delegating} and \citet{saig2023delegated} study how to delegate the training of a model. The key distinction from our model is that they both assume that the principal (i.e., the party delegating the task) is the contract designer. 
In particular,  \citet{ananthakrishnan2023delegating} study how to incentivize an agent to collect samples used to train a classifier, and \citet{saig2023delegated} study a similar problem under budget constraints.

%% file: src/prelim.tex
\section{Model}\label{sec:prelim}

An instance of the \emph{delegation problem} is characterized by a tuple $\delegation\defeq (\Theta,\vxi,\Omega, \cA, \mF, \mR, \vc)$, where: $\Theta$ is a finite set of $n$ user's types and $\vxi\in\Delta(\Theta)$ is the distribution of users over such types; $\Omega$ is a finite set of $m$ outcomes; $\cA$ is a finite set of $\ell\coloneqq|\cA|$ actions available to the service provider; $\mF$ is an $m\times \ell$ left stochastic matrix specifying for each action $a \in\cA$ a probability distribution over the outcomes, that is each column is such that $\mF_a\in\Delta(\Omega)$ for all $a\in\cA$; $\mR$ is an $m\times n$ matrix specifying the user's reward for each pair outcome/user's type, \ie each columns $\mR_\theta\in[0,1]^m$ is the reward of an user of type $\theta\in\Theta$ for all $m$ outcomes; finally, the $\ell$-dimensional vector $\vc$ specifies the cost incurred by the service provider in performing each action, that is $c_a$ is the cost of performing action $a$.\footnote{In this work, for any finite set $\cS$, we denote with $\Delta(\cS)$ the set of possible distributions over $\cS$. For any $k\in\naturals$, we denote with $\range{k}$ the set $\{1,\ldots, k\}$.}

\subsection{Model of interaction and optimization problem}
In this work, the goal of the service provider is to maximize its own utility by designing a suitable payment mechanism for the user. We start by considering the following \emph{deterministic payment mechanism} implemented by a \emph{menu} of \emph{deterministic payment schemes}. A menu of deterministic payment schemes $\PS=\{(a_i,\vp_i)\}_{i\in\range{k}}$ is characterized by $k\in\naturals$ options presented to the user, where each option corresponds to an action $a_i\in\cA$ and a payment vector $\vp_i\in\reals^{m}_{\ge 0}$, defining a payment for each of the $m$ outcomes. We refer to each tuple $(a,\vp)\in\PS$ as a \emph{payment scheme}. Moreover, the user can also decide to disregard all presented options and decide to \emph{opt-out} from the mechanism. The user's choice to opt-out is denoted by the symbol $\varnothing$. We refer to the carnality $k$ of a menu $P\in\cP_k$ as \emph{menu complexity}.
The set of all deterministic menu payment schemes of complexity $k$ is denoted with $\cP_k$, while the set menus of deterministic payment schemes is $\cP=\cup_{k\in\mathbb{N}}\cP_k$.
In \Cref{sec:randomized} we will extend the mechanism to include menus of \emph{randomized} payment schemes.

The interaction between the user and the service provider proceeds as follows:

\begin{enumerate}[label={\roman*)}] 
	\item The {service provider} commits to a menu of deterministic payment schemes $\PS=\{(a_i, \vp_i)\}_{i\in \range{k}}$;
	\item The user's type $\theta\in\Theta$ is drawn according to the type distribution $\vxi$. The type $\theta$ is hidden from the {service provider}; 
	\item The user selects a payment scheme $(a, \vp)\in\PS$ or the opt-out option $\varnothing$;
\end{enumerate}
	 If the user selected a payment option $(a, \vp)$ at step iii) then:
	\begin{enumerate}[start=4, label={\roman*)}] 
	\item The {service provider} performs action $\hat a\in \cA$ and incurs in cost $c_{\hat a}$. The action $\hat a$ is hidden from the user. In principle, it may be the case that $\hat a\neq a$;

	\item An outcome $\omega$ is sampled according to the outcome distribution $F_{\hat a}$;

	\item The {service provider} receives payment of $p(\omega)$ from the user.

\end{enumerate}

\noindent Alternatively, if the user has chosen the opt-out option $\varnothing$, no further action occurs, resulting in zero utility for both the {service provider} and the user.

\xhdr{User's behavior.} 
By following the above interaction, the user's expected utility for selecting a payment scheme $(a,\vp)\in\PS$, assuming the service provider executes action $a$, can be written as $\mF_{a}^\top (\mR_\theta-\vp)$. The corresponding expected utility of the service provider can be computed as $\mF_{a}^\top\vp-c_{a}$.
Given a menu of deterministic payment scheme $\PS\in\cP_k$, a user of type $\theta$ selects the payment scheme that maximizes their own expected utility. Formally, we specify the option selected by the user through a selection function $\ui:\Theta\to\range{k}\cup\{\varnothing\}$, which is such that 
\begin{equation}\label{eq:selection func}
\ui(\theta\mid \PS)\defeq  \mleft\{\hspace{-1.25mm}\begin{array}{l}
	\displaystyle
	\varnothing \hspace{4cm} \textnormal{if }\,\, \mF_{a_i}^\top (\mR_\theta-\vp_i)< 0 \textnormal{ for all } i\in\range{k}\\[4mm]
	\displaystyle j \in \argmax_{ i\in I(\theta|\PS)}\left\{\mF_{a_i}^\top \vp_i-c_{a_i}\right\}\hspace{.5cm}\textnormal{otherwise}
\end{array}\mright.
\end{equation}
where $I(\theta\mid\PS)\coloneqq  \arg\max_{i\in\range{k}} \left\{\mF_{a_i}^\top (\mR_\theta-\vp_i)\right\}$, meaning that ties are resolved in favor of the {service provider}. Therefore, when none of the $k$ deterministic contracts prove to be profitable for the user, the selection function chooses $\varnothing$; otherwise, it returns one of the options that maximizes the expected user's utility. To simplify notation, we omit the dependence of the selection function on the payment scheme $\PS$ when it is clear from the context.

\xhdr{IC constraints.} 
When describing the user's best response, we assumed that the {service provider} behaves ``truthfully'' by performing the action $a\in\cA$ corresponding to the payment scheme $(a_i,\vp_i)$ selected by the user. However, the {service provider} could have incentives to play a different action $\hat a\neq a_i$ when payment scheme $(a_i,\vp_i)\in\PS$ is selected. Given that this realized action $\hat a$ remains concealed from the user, who has only partial information on the actions undertaken by the service provider, one important feature of each payment scheme is ensuring the user can trust the platform when delegating tasks. 
This trust is established by ensuring that the {service provider} is incentivized to perform the action corresponding to the user's ``contractual'' choice. As common in the literature, we refer to these constraint as {service provider}'s \emph{incentive compatibility constraints} (\icsp). We say that the payment scheme $(a,\vp)$ is IC if the action $a$ is a best response to the corresponding payment $\vp$. More precisely, we say that a menu of deterministic contracts is IC if, for all $i\in\range{k}$, we have that $a_i\in\argmax_{a\in \cA}\left\{\mF_{a}^\top\vp_i-c_a\right\}$.

\subsection{Optimization Problems under Deterministic Payment Schemes}\label{sec:opt prog}

First, we provide a description of the optimization problem for computing an optimal menu of deterministic payment schemes of a fix menu complexity $k$. Thus, initially, we are thinking about $k$ as a parameter of the problem, given by external constraints, not under the control of the service provider. In later sections we will discuss the dependency of the utility of menus of deterministic payments schemes as a function of $k$. 
Second, we describe the optimization problem for computing optimal direct menus of deterministic payment schemes. 

\xhdr{Optimal Menu of Deterministic Payment Schemes.} 
Now we describe the optimization problem faced by the service provider that has to find optimal menu of deterministic payment scheme of fixed complexity $k$.
To simplify the exposition --- with slight abuse of notation--- we introduce symbols $\vp_{\varnothing}=\vzero$, $c_{a_\varnothing}=0$, and $\mF_{a_\varnothing}=\vzero$.
Given an instance $(\Theta,\vxi,\Omega, \cA, \mF, \mR, \vc)$ of the delegation problem, and given a size $k$ for the menu, the {service provider}'s goal is to design a menu of $k$ contracts which is incentive compatible and maximizes its own utility:
\begin{equation}\label{prog:determ}
	\OPT_k\coloneqq
	\mleft\{\begin{array}{cll}
		\displaystyle
		\max_{\PS=\{(a_i,\vp_i)\}_{i\in\range{k}}\in\cP_k} &\displaystyle \sum_{\theta\in\Theta}\xi_\theta\left(\mF_{a_{\ui(\theta\mid\PS)}}^\top \vp_{\ui(\theta\mid\PS)}-c_{a_{\ui(\theta\mid\PS)}}\right) &  \\[5mm]
		\textnormal{s.t.}&   \mF_{a_i}^\top \vp_i-c_{a_i} \ge \mF_{a'}^\top \vp_i-c_{a'}& \forall a'\in\cA, \forall i\in\range{k}
	\end{array}\mright.
\end{equation}

\noindent 
We observe that in the aforementioned program, we did not explicitly impose any \emph{individual rationality} (IR) constraints on the {service provider}, as these are implicitly guaranteed by the fact that the payment mechanism is designed by the {service provider} itself. We refer to \Cref{app:IR} for a discussion on this.

\xhdr{Direct Menus of Deterministic Payment Schemes.}
An important class of menus of deterministic payment schemes is the one of \emph{direct} menus of deterministic payment schemes. Through standard revelation-principle-style arguments, it can be shown that the optimal value achievable through a menu of deterministic payment schemes can always be attained by a menu of complexity $k=n$ that is both direct and incentive compatible. 
This class is more powerful then the one of fixed complexity $k$ in terms of utility that can be achieved (we will discuss in \Cref{sec:classes} the fraction of service provider utility that can be harvested when the size $k$ is fixed.), however they can be reasonably employed in practice only when the number of types $n$ is small, as thus they have limited applicability in practice.

A direct menu can be tough of a menu in which each pair $(a,\vp)\in\PS$ as being associated to one specific type in $\Theta$. Then, a user of type $\theta$ has a corresponding deterministic payment scheme $(a_\theta,\vp_\theta)$, and this payment scheme is preferred by the user over the payments schemes $(a_{\theta'},\vp_{\theta'})$ associated to other types $\theta'\ne\theta$ and over the out-out option. 
Formally, let
\[
	\cM\defeq (\cA\times \R^m_{\ge 0}) \cup \{(a_\varnothing,\vp_\varnothing)\}
\]
be the set of possible deterministic payments schemes, extended with the fictitious pair $(a_\varnothing,\vp_\varnothing)$, corresponding to the case in which the service provider suggests to the user to abstain from participating in the mechanism at a cost of 0. Then, the service provider can compute an optimal direct menu of deterministic payment schemes as follows

\begin{subequations}\label{prog: optimal direct deterministic}
	\begin{empheq}[left=\OPT\defeq\empheqlbrace]{align}
		\max\limits_{\PS=\{(a_\theta,\vp_\theta)\in\cM\,:\,\theta\in\Theta\}} \hspace{.2cm}&\sum_{\theta\in\Theta}\xi_\theta\left(\mF_{a_{\theta}}^\top \vp_{\theta}-c_{a_{\theta}}\right)\nonumber\\
		\textnormal{s.t.}\hspace{1.3cm} & \mF_{a_\theta}^\top \vp_\theta-c_{a_\theta} \ge \mF_{a'}^\top \vp_\theta-c_{a'} &\forall a'\in\cA, \forall \theta\in\Theta,\label{eq:ICsp2}\\
	   \phantom{s.t.}\hspace{1cm}& \mF_{a_\theta}^\top (\mR_\theta-\vp_\theta) \ge \mF_{a_{\theta'}}^\top (\mR_{\theta}-\vp_{\theta'})&\forall \theta,\theta'\in\Theta,\label{eq:ICuser}\\
	   \phantom{s.t.}\hspace{1cm}& \mF_{a_\theta}^\top (\mR_\theta-\vp_\theta) \ge 0&\forall\theta\in\Theta.
	\end{empheq}
\end{subequations}

\noindent We observe that IC constraints \eqref{eq:ICuser} have an interpretation in terms of the choice function $\ui:\Theta\to\range{n}\cup\{\varnothing\}$. Indeed \Cref{eq:ICuser} can be stated as requiring $\ui(\theta\mid\PS)$ selects $\theta$ for each type $\theta$.

%% file: src/NPHard.tex
\section{Computing the optimal direct menu is NP-Hard}\label{sec:hardness}

In this section we present a negative result regarding the problem of finding an optimal direct menu of deterministic payment schemes. Formally, we show that the optimal menu of deterministic payment schemes cannot be approximated within any constant multiplicative factor. Then, the problem is \emph{not} in $\mathsf{APX}$. We prove this result by reduction from \IND. In particular, we use that, for any $\epsilon>0$, it is \NPHARD~to distinguish between an instance of \IND~in which the maximum independent set is of size $|V|^{1-\epsilon}$, and one in which all independent sets are of size $|V|^\epsilon$~\cite{hastad1999clique,Zuckerman2007linear}. 

\begin{theorem}\label{thm:negative}
	For every constant $\rho>0$, there exists no polynomial-time algorithm that approximates up to within a $\rho$ multiplicative factor the value of the service provider's expected utility under an optimal (direct) menu of deterministic payment schemes, unless $\mathsf{P}=\mathsf{NP}$.
\end{theorem}

%% file: src/single.tex
\section{Computing optimal menus with small complexity}\label{sec:pricing}

In this section, we present an algorithm for computing an optimal menu of deterministic payment schemes (Problem \eqref{prog:determ}) that runs in polynomial time when the menu complexity is ``small'', that is, the complexity of the menu is small, \ie fixed to a constant $k=O(1)$. 
In order to do that, we reduce the problem to a particular multi-item pricing problem with unit-demand and in which the price of each item has to be greater than or equal to a specific \emph{price-floor} (Problem \eqref{prog:pricing}). We show how to solve the new problem in polynomial time when $k=O(1)$, and how to recover an optimal solution to Problem \eqref{prog:determ} in polynomial time.

\subsection{The Unit-Demand Pricing Problem with Price-Floors}
We start by characterizing the set of \emph{feasible expected payments}, which will be essential in the pricing formulation. Indeed, this set let us encode the IC constraints, associated with the service provider and thus let us reduce to the price problem briefly discussed above. More precisely is this set that introduces the price floors on the items in the pricing view of our problem.

\begin{definition}[Feasible expected payments]\label{def:Q}
	Given an instance $\Gamma$ of the delegation problem, for each action $a\in \cA$ we define the set $Q_a\subseteq \reals_{\ge0}$ as the set of expected payments satisfying the service provider's IC constraints for action $a$ as per \Cref{eq:ICsp2}.
	Formally, 
	\[
	Q_a\defeq\left\{q=\mF_a^\top\vp:\, \mF_a^\top\vp-c_a\ge \max\limits_{a'\in\cA}\mF_a^\top\vp-c_{a'}\right\}.
	\]
\end{definition}

\noindent
Clearly, $Q_a$ is a polytope for every $a\in \cA$, since it is a linear transformation of the polytope of IC payments, which is defined as
\[
\Pi_a\defeq\left\{\vp:\, \mF_a^\top\vp-c_a\ge \max\limits_{a'\in\cA}\mF_{a'}^\top\vp-c_{a'}\right\}\subseteq \reals_{\ge 0}^m,
\]
for each action $a\in\cA$.
Then, since $Q_a$ is $1$-dimensional, it is an interval. Next, we show that we can further characterize $Q_a$. In particular, $Q_a$ is an interval $[l_a, +\infty)$ unbounded above. Intuitively, this signifies that an action of the service provider can be induced by all the expected payment greater than $l_a$.

\begin{restatable}{lemma}{shapeofQ}
	For each action $a\in \cA$, it holds $Q_a=[l_a, +\infty)\subseteq \reals_{\ge 0}$ for some $l_a\ge 0$.
\end{restatable}

Now, we leverage the definition of the sets $Q_a$ to formulate a new problem. Solving this problem allows us to recover a solution for the original problem described in Program \eqref{prog:determ}, i.e.,  the problem of computing an optimal menu of size $k$ of deterministic payment schemes. For simplicity in notation, we let $q_\varnothing\defeq 0$. The new optimization problem reads as follows:
\begin{equation}\label{prog:pricing}
\mleft\{
\begin{array}{cll}
	\max\limits_{(a_i, q_i)_{i\in\range{k}}} &\displaystyle\sum_{\theta\in\Theta} \xi_\theta\left(q_{\ui(\theta)}-c_{a_{\ui(\theta)}}\right)\\
	\textnormal{s.t.}&\mF^\top_{a_{\ui(\theta)}}\mR_\theta-q_{\ui(\theta)}\ge \mF^\top_{a_{j}}\mR_{\theta}-q_{j},&\forall\theta\in\Theta, j\in\range{k}\cup\{\varnothing\}\\
	& q_{i}\in Q_{a_i},&\forall i\in\range{k}\\
	& a_i\in\cA, &\forall i\in\range{k}
\end{array}
\mright.
\end{equation}
where the user's selection function $i(\cdot)$ is defined as per \Cref{eq:selection func}. First, we observe that the new problem has just $O(k)$ variables, as opposed to the $O(km)$ variables of Problem \eqref{prog:determ}. Upon examining the program presented above, one can observe similarities with a closely related \emph{unit-demand pricing problem}. Indeed, Problem \eqref{prog:pricing} can be interpreted as follows: 
\begin{enumerate}[label={\roman*)}] 
\item There is a single buyer of type $\theta$, which is hidden from the {service provider}, and drawn from the type distribution $\vxi$;
\item The seller sets a price $q_i\in Q_{a_i}$ for each item $i\in\range{k}$;
\item The buyer has a unit demand and has utility $\mF_{a_i}^\top \mR_\theta$ for the $i$-th item;
\item The buyer selects a single item $\ui(\theta)$ that maximizes their utility;\footnote{We allow for a slight abuse of notation by letting $\ui:\Theta\to\range{k}\cup\{\varnothing\}$ as in \Cref{eq:selection func}, since the buyer's behavior in the new pricing problem is the same of the user's behavior described in \Cref{eq:selection func}. Indeed, $\ui(\theta)$ is the index of an item with the highest value for a buyer of type $\theta$. When items have equal value for the buyer, ties are resolved in favor of the seller.}
\item The seller produces the item selected by the buyer, and pays the production cost $c_{a_{\ui(\theta)}}$.
\end{enumerate}
We note that, under this interpretation, each item $i$ is associated with a \emph{price-floor} $l_{a_i}$, which gives a lower bound on the magnitude of the feasible prices.

\xhdr{Equivalence of Problem \eqref{prog:determ} and \eqref{prog:pricing}.}
In the following, we will prove that Problem \eqref{prog:determ} and Problem \eqref{prog:pricing} are equivalent to each other. Indeed, we will prove that, in polynomial time, one can build a solution for one problem from a solution to the other with the same objective value.

\begin{restatable}{lemma}{lemmapricingtodeterm}\label{lem:pricingtodeterm}
	Given a feasible solution $\{(a_i, q_i)\}_{i\in\range{k}}$ to the pricing Problem \eqref{prog:pricing}, it is possible to construct in polynomial time a feasible solution $\PS=\{(a_i, \vp_i)\}_{i\in\range{k}}$ for Problem \eqref{prog:determ} with the same objective value.
\end{restatable}

\noindent Similarly, it is possible to show that the converse is also true.

\begin{restatable}{lemma}{lemmapaymenttopricing}\label{lem:payment to pricing}
	Given a feasible solution $\{(a_i, \vp_i)\}_{i\in\range{k}}$ to Problem \eqref{prog:determ}, it is possible to construct in polynomial time a feasible solution $\{(a_i, q_i)\}_{i\in\range{k}}$ to Problem \eqref{prog:pricing} with the same value.
\end{restatable}

The two lemmas we just described show that Problem \eqref{prog:pricing} reduces to Problem \eqref{prog:determ}, and vice versa. 
In particular, in the upcoming discussion we will leverage the following result, which can be directly derived from the preceding lemmas. 
\begin{corollary}\label{cor:from pricing to contract}
	Given a polynomial-time algorithm for Problem~\eqref{prog:pricing}, we can design a polynomial-time algorithm for Problem \eqref{prog:determ}.
\end{corollary}

Now we will turn our attention to build such a polynomial algorithm for the pricing problem with unit-demand and price-floors.

%% file: src/single2.tex
\subsection{Polynomial-Time Algorithm for Optimal Contracts with Small Menu Complexity} 
We present an algorithm for Problem \eqref{prog:pricing} that runs in polynomial time when the menu complexity is a constant $k=O(1)$. The algorithm exploits the existence of an optimal solution to the pricing problem which is an extreme point of the constraints set of Problem \eqref{prog:pricing}. The crux of the algorithm is leveraging the low dimensionality of the auxiliary pricing problem. Indeed, the reformulation allows us to consider only $k$ hyperplanes at a time when searching for an optiaml solution among the extreme points of the feasible set. This is possible since the dimensionality of Program \eqref{prog:pricing} is $k$ when the choices of actions $a_i\in\cA$ for each $i\in\range{k}$ are fixed. The same approach would not be feasible under the original formulation of the problem (see Problem \eqref{prog:determ}), since the dimensionality of the problem would be $km$, and the total number of hyperplanes to consider would be exponential also in the number of outcomes $m$.

Let $\va=(a_1,\ldots,a_k)$ be a fixed choice of actions, one for each of the $k$ payment schemes in the menu. For a fixed choice of actions $\va$, Problem \eqref{prog:pricing} is a linear programming problem with variables $\vq=(q_1,\ldots,q_k)\in\reals^k_{\ge 0}$.
Consider the first familiy of constraints in Problem \eqref{prog:pricing} for a fixed $\va$. For any $i,j\in\range{k}$ and type $\theta\in\Theta$, the corresponding constraint is defined by an hyperplane 
\begin{equation}\label{eq:hyper1}
	H^{i, j, \theta}_{\va}\coloneqq\left\{\vq\in\reals^k:\, q_i-q_j=(\mF_{a_i}-\mF_{a_j})^\top\mR_\theta\right\}.
\end{equation}
Moreover, for $i\in\range{k}$, the second famility of constraints of Problem \eqref{prog:pricing} is determined by an hyperplane 
\begin{equation}\label{eq:hyper2}
	\tH^{i}_{\va}\coloneqq\left\{ \vq\in\reals^k:\,q_i=l_{a_i}\right\}.
\end{equation}

\Cref{alg:pricing} outlines the primary steps of the procedure used to address Problem \eqref{prog:pricing}. The algorithm iterates over the possible combinations of actions $\va\in\cA^k$. For each $\va$, the algorithm iteratively builds a set of candidate points $\cF_{\va}\subseteq \reals^k_{\ge 0}\cup\{\varnothing\}$. Then, these sets are used to compute the final solution by maximizing the original objective of the pricing problem. For each tuple $\va$, the algorithm initializes $\cF_{\va}$ to contain the opt-out option $\varnothing$. Then, the algorithm enumerates over the possible combinations of $k$ linearly independent hyperplanes from the set $\cH_{\va}$, which is the union of all hyperplanes defined according to \Cref{eq:hyper1} and \Cref{eq:hyper2} given $\va$. For each combination of $k$ linearly independent hyperplanes, the algorithm computes a candidate point $\vq\in\reals^k$ as the intersection of those $k$ hyperplanes. If the resulting point is $\vq\ge 0$, then the algorithm adds it to $\cF_{\va}$. 

The following result shows that when $k=O(1)$ the algorithm succesfully computes an optimal solution to the unit-demant pricing problem with price-floors in polynomial time.

\begin{algorithm}
	\begin{algorithmic}
	\caption{Solving the unit-demand pricing problem with price-floors}\label{alg:pricing}
	\FOR{$\va=(a_1,\ldots,a_k)\in\cA^k$}{
		\STATE Initialize $\cF_{\va}\gets\{\varnothing\}$
		\STATE Define the set of hyperplanes: 
		\(\cH_{\va}\coloneqq \left\{H^{i,j,\theta}_{\va}\right\}_{i,j\in\range{k},\theta\in\Theta}\cup\left\{\tH^{i}_{\va}\right\}_{i\in\range{k}}\)
		\FOR{every set $H_1,\ldots,H_k$ of $k$ linearly independent hyperplanes from $\cH_{\va}$}{
			\STATE Compute $\vq=\bigcap_{i\in\range{k}}H_i$
			\IF{$\vq\in\reals^k_{\ge0}$}
			{
				\STATE $\cF_{\va}\gets \cF_{\va}\cup\{\vq\}$
			}
			\ENDIF
		}
		\ENDFOR
	}
	\ENDFOR
	\RETURN{\(
		\{(a_i, q_i)\}_{i\in\range{k}}\in\argmax_{\va\in\cA^k, \vq\in \cF_{\va}}\,\,\sum_{\theta\in\Theta}\xi_\theta\left(q_{\ui(\theta)}-c_{a_{\ui(\theta)}}\right)\)}
		\end{algorithmic}
\end{algorithm}

\begin{restatable}{theorem}{polyalgo}\label{thm:pricing polytime}
	\Cref{alg:pricing} finds an optimal solution to Problem \eqref{prog:pricing} in polynomial time whenever $k=O(1)$.
\end{restatable}

The main result of this section follows from a direct application of \Cref{thm:pricing polytime} and \Cref{cor:from pricing to contract}.

\begin{theorem}\label{th:opt_k}
	If $k$ is a constant, the problem of finding a menu of size $k$ of deterministic payment schemes with utility $\OPT_k$ can be solved in polynomial time.
\end{theorem}

This result is extremely important as it allows the service provider to compute optimal menus in most cases of real interest. Indeed, many situations require a limited and sufficiently small contract complexity $k$, since these are the options presented to users of the service. For instance, as a direct corollary, we have that the optimal single contract, \emph{i.e.}, menu of size $k=1$, can be found in polynomial time.

In general, the optimal utility is reached through direct menus, that is, when $k=n$ and each type has a dedicated payment scheme (see Problem \eqref{prog: optimal direct deterministic}). This observation yields the following corollary:
\begin{corollary}
	When the number of types $n$ is constant, there exists a polynomial time algorithm to compute a menu of deterministic payment schemes with utility $\OPT$.
\end{corollary}
\noindent This is not the only case in which an optimal menu can be computed efficiently. Indeed, if the number of actions is smaller then the number of types, one can build menus of deterministic payment schemes with complexity $k=\ell:=|\cA|$ that give utility $\OPT$ to the service provider.
\begin{restatable}{lemma}{lemmaminactions}\label{lem:smallactions}
	Given an instance of the delegation problem $\delegation$ with $\ell\le n$, there always exits a menu of complexity $k=\ell$ with service provider's utility $\OPT$, \ie $\OPT_\ell=\OPT$.
\end{restatable}

\noindent \Cref{lem:smallactions} and \Cref{th:opt_k} yield the following corollary:
\begin{corollary}\label{cor:constantactions}
	When the number of actions $\ell$ is constant, there exists a polynomial time algorithm to compute a menu of deterministic payment schemes with utility $\OPT$.
\end{corollary}

In the next section, we characterize how much utility can be extracted by a menu of complexity $k$.

%% file: src/powerOfClasses.tex
\section{Trade Off between Utility and Complexity of Menus}\label{sec:classes}

In this section, we study the dependency of the {service provider}'s utility on the complexity $k$ of the proposed menu of deterministic payment schemes, with respect to the utility collected by an optimal direct menu solving Problem \eqref{prog: optimal direct deterministic}. We show that a menu of complexity $k\le \min(n,\ell)$ can extract a $k/\min(n,\ell)$ fraction of the total utility in the worst case, and that this bound is tight.
Interestingly, this factor does not depend on the number of outcomes $m$ of the delegation instance.

We start by providing a lower bound on $\OPT_k$ in terms of the fraction of $\OPT$ that it can guarantee in the worst case. The former quantity is the expected utility for a menu of size $k$ obtained as the solution to Problem \eqref{prog:determ}, and the latter is the optimal objective value of Problem \eqref{prog: optimal direct deterministic}.
\begin{restatable}{theorem}{theoremsingleGood}\label{prop:singleGood}
	Given an instance of the delegation problem $\delegation$, the expected utility of the {service provider} obtained through an optimal menu of complexity $k$ is at least a $k/\min(n,\ell)$ fraction of the expected utility of the optimal direct menu of complexity $n$, \ie \[\OPT_k\ge \frac{k}{\min(n,\ell)}\OPT.\]
\end{restatable}

Crucially, this result shows that the menu complexity is the only parameter that controls the utility of the service provider. This 

Now we show that $k/n$ is tight in terms of the worst-case maximum utility that the {service provider} can extract by using a menu of size $k$, \ie that there are instances in which all menus of size $k$ can extract \emph{at most} a $k/\min(n,\ell)$ fraction of the optimal utility.

\begin{restatable}{proposition}{propositionsingleBad}\label{prop:singleBad}
	For each $k\ge 0$ and $n\ge 0$, there exists an instance with $n$ types and actions such that the  expected utility of the {service provider} by using a menu of deterministic payments schemes of size $k$ yields most a $k/n$ fraction of the optimal direct menu of size $n$, \ie $\OPT_k\le \frac{k}{n}\OPT$.
\end{restatable}

%% file: src/monotonenew.tex
\section{Continuous actions and robustness}\label{sec:continuous}

In the previous sections, we focused on scenarios with finite sets of actions. However, numerous real-world problems feature action sets that are continuous. For instance, this is the case in the earlier motivating example of delegating the training of a machine learning model. In this case, actions could represent the time allocated for training, which is modeled naturally with a continuous variable. Moreover, the action set may be so large to the extent that, for all practical purposes, it can be regarded as continuous. Therefore, the goal of the section is to address cases in which the action set is $\mathcal{A} = [0,1]$.

We assume the following mild and natural regularity condition, requiring the dependence of the outcome distribution and costs on the service provider's action to be sufficiently smooth (\ie outcome distributions and costs for similar actions should be similar).
\begin{assumption}[Action regularity]\label{ass:Lip}
	The outcome distribution is $L_{\mF}$-Lipschitz continuous with respect to action $a$, that is, for any $a,a' \in \mathcal{A}$:
	\[
	\|\mF_a-\mF_{a'}\|_1\le L_{\mF}|a-a'|.
	\]
	Similarly, the costs are $L_{c}$-Lipschitz continuous with respect to action $a$: for all $a,a' \in \mathcal{A}$,
	\[
	|c_a-c_{a'}|\le L_{c}|a-a'|.
	\]
\end{assumption}
\noindent
To simplify the exposition, we assume that $\mF$ and $c$ are both $1$-Lipschitz continuous.

Following a reasoning analogous to that of \Cref{sec:pricing}, it is possible to rewrite the problem of computing an optimal direct menu of deterministic payment schemes as a unit-demand pricing problem with price-floors. In particular, we obtain the following optimization problem:
\begin{equation}\label{prog:pricingcnt}
\mleft\{\begin{array}{cll}
	\sup\limits_{\{(a_\theta, q_\theta)\}_{\theta\in\Theta}} &\sum_{\theta\in\Theta} \xi_\theta\left(q_{\theta}-c_{a_{\theta}}\right) \\
	\textnormal{s.t.} & \mF^\top_{a_{\theta}}\mR_\theta-q_{\theta}\ge \mF^\top_{a_{\theta'}}\mR_\theta-q_{\theta'}&\forall\theta,\theta'\in\Theta\\
	&\mF^\top_{a_{\theta}}\mR_\theta-q_{\theta}\ge0&\forall\theta\in\Theta\\
	& q_{\theta}\in Q_{a_{\theta}}&\forall \theta\in\Theta\\
	& a_\theta\in[0,1]&\forall\theta\in\Theta
\end{array}\mright.
\end{equation}
While this program seems similar to the one solved in \Cref{sec:pricing}, it hides many non-trivial challenges. %

\xhdr{Technical Challenges and Solution.}
A natural approach to solve the problem with continuous actions is to discretize the action set $\cA$, and solve such a problem as seen in \Cref{sec:pricing}. However, such direct approach alone is not sufficient to guarantee a good approximation of the service provider's utility in the original continuous problem. 
Indeed, it is easy to see that, considering slight modifications of the actions would lead to significant consequences in the service provider's utility.
Therefore, we need to design a suitable algorithm to solve the continuous-action problem. We propose an algorithm which is composed of two main steps. First, we solve a discretized version of the continuous problem under \emph{relaxed} user's IC and IR constraints.
Then, we prove that this solution gives an expected utility to the service provider which is close to the value of an optimal solution of the problem with continuous actions. However, this solution only approximately satisfies IC and IR constraints for the user. Therefore, we need to suitably adjust the solution of the discretized and relaxed problem, so that the user's IC and IR constraints are fully satisfied.
We achieve this goal by a clever ``reimbursement'' technique which slightly decreases the magnitude of payments, and pays back the user a small fraction of the service provider's earnings.
As we discuss in the following section, this procedure allows the service provider to build \emph{robust} contracts.

\subsection{Computing Robust Contracts}

In order to provide a solution for continuous-action settings, we will first address the following problem: the service provider would like to compute an optimal contract $\PS$, but they can only rely on estimates $\tR$ of the user's reward matrix $\mR$. %
The true reward matrix $\mR$ is unknown to the service provider. 
We show that, given an IC and IR (both for the user and the service provider) menu $\tPS$ computed over an estimated $\tR$ such that $\|\tR - \mR\|_{\infty}\le \delta$, we can build a new menu achieving an expected service provider's utility close to that extracted by $\tPS$, while guaranteeing a small violation in the {service provider}'s incentive compatibility. The interesting part of these guarantees is that they apply even if the service provider cannot exactly predict the behavior of the user, since the user is best-responding according to an unknown reward function.

We need to introduce a few concepts to handle approximate rationality constraints for both the service provider and the user.
First, we introduce a relaxed version of the set $Q_a$. For any $a\in\cA$ and $\epsilon\ge0$, let $Q_a^\epsilon$ be the set of expected payments that ``$\epsilon$-induce'' action $a$. Formally:
\[
Q^\epsilon_a\defeq\left\{q=\mF_a^\top\vp:\, \mF_a^\top\vp-c_a\ge \max\limits_{a'\in\cA}\mF_{a'}^\top\vp-c_{a'}-\epsilon\right\}.
\]
This allows us to formalize approximate IC constraints for the service provider.
\begin{definition} Given any $\epsilon>0$, we say that a direct menu of deterministic payment schemes $P\coloneqq\{(a_\theta,q_\theta)\}_{\theta\in\Theta}$ is $\epsilon$-IC for the service provider if $q_{\theta}\in Q_{a_\theta}^\epsilon$ for all $\theta$.
\end{definition}
\noindent The next result shows that $Q_\epsilon^a$ is ``well-behaved'' in $\epsilon$, meaning that if we subtract a quantity $\epsilon'$ to an expected payment that was $\epsilon$-inducing action $a$, then we obtain an expected payment that $(\epsilon'+\epsilon)$-induces action $a$.
\begin{restatable}{lemma}{lemmaQeps}\label{lem:Qeps}
	For any $\epsilon\ge0$, $a \in \cA$, $q\in Q_a^\epsilon$, and $0\le\epsilon'\le q$, it holds $(q-\epsilon')\in Q_a^{\epsilon+\epsilon'}$. 
\end{restatable}
Now, we turn our attention to the approximate incentives of the user, and we define a relaxed versions of its IC and IR constraints.
\begin{definition}
Given $\delta\ge 0$, we say that a direct menu $\PS:=\{(\tilde a_\theta, \tilde q_\theta)\}_{\theta\in\Theta}$ is $\delta$-IC for the user if:
\[
\mF^\top_{\tilde  a_\theta}\mR_\theta-\tilde q_\theta\ge\mF^\top_{\tilde a_{\theta'}}\mR_\theta-\tilde  q_{\theta'}-\delta\quad\forall\theta,\theta'\in\Theta,
\]
and $\delta$-IR for the user if:
\[
\mF^\top_{\tilde a_\theta}\mR_\theta-\tilde  q_\theta\ge-\delta \quad\forall\theta,\theta'\in\Theta.
\]
\end{definition}
The following result shows that it is possible to build menus which are robust to small variations in the user's utility matrix. In particular, suppose we compute a menu $\tPS$ for $\tR$, but the true underlying matrix is $\mR$ such that $\|\tR-\mR\|_{\infty}\le \delta$. Then, it is easy to see that the menu $\tPS$ computed over $\tR$ would be $\delta$-IC and $\delta$-IR with respect to the true constraints given by $\mR$.
However, it is undesirable that the burden of uncertainty falls on the user, as we, as the designer of the mechanism, cannot control their behavior. Thus, we show how, in the presence of uncertainty, we can move the approximate rationality from the user to the service provider.
More formally, we show that, starting from $\tPS$, we can build a new menu which is almost IC for the service provider and which guarantees a good service provider's utility against a user who behaves according to the unknown rewards. Therefore, $\tPS$ is robust with respect to the strategic behavior of the user which cannot be fully anticipated by the service provider, since the user is best-responding according to the unknown reward matrix $\mR$.

\begin{restatable}{theorem}{lemmadeltaIC}\label{th:deltaIC}
	Let $\tPS:=\{(\tilde a_\theta, \tilde \vp_\theta)\}_{\theta\in\Theta}$ be an $\epsilon$-IC menu for the service provider, \ie $\mF^\top_{\tilde a_\theta}\tilde\vp_{\theta} \eqqcolon \tilde q_\theta\in Q^\epsilon_{\tilde  a_\theta}$ for each $\theta$, and $\delta$-IC and $\delta$-IR for the user. Let $\widetilde{\OPT}$ be the expected service provider's utility extracted by $\tPS$. 
	Then, there exist a polynomial time algorithm that computes an indirect menu $\PS$ that is a function only of $\tPS$ and $\delta$ such that
	\begin{OneLiners}
		\item the utility extracted by $\PS$ is greater than or equal to $\widetilde{\OPT}-O(\sqrt{\delta})$;
		\item $\PS$ is $O(\epsilon+\sqrt{\delta})$-IC for the service provider.
	\end{OneLiners}
\end{restatable}

\noindent 
We observe that the ``robust'' menu $P$ is indirect. Indeed, a direct menu cannot be computed since the service provider is unaware of the parameters of the instance (\eg the true reward matrix $\mR$), making it impossible to predict which payment scheme will be chosen by each user's type.
However, despite being indirect, the service provider can still safely deploy the robust menu since it is $0$-IR and $0$-IC by definition, as the users are free to select whichever payment scheme they prefer according to their true utilities.

Along the same line, note that the property that $P$ is independent of $\mR$ is fundamental to design robust contracts. Indeed, from \cref{th:deltaIC} it directly follows:

\begin{corollary}
	Consider an IC and IR (both for the user and the service provider) menu $\tPS$ computed over an estimated $\tR$ such that $\|\tR - \mR\|_{\infty}\le \delta$. Let $\widetilde{\OPT}$ be the expected service provider's utility extracted by $\tPS$. 
	Then,  there exist a polynomial time algorithm that computes an non-direct menu $\PS$ that is a function only of $\tPS$ and $\delta$ such that
		\begin{OneLiners}
		\item the utility extracted by $\PS$ is greater than or equal to $\widetilde{\OPT}-O(\sqrt{\delta})$;
		\item $\PS$ is $O(\sqrt{\delta})$-IC for the service provider.
	\end{OneLiners}
\end{corollary}

\subsection{Continuous Actions: Discretizing and Relaxing}

The results on computing robust menus play a key role in the computation of menus for the case of continuous actions. In order to address this problem, we start by making the following natural assumption on the structure of outcome probabilities. 

\begin{definition}[$c$-smoothness]
	An instance of the delegation problem $\delegation$ is $c$-smooth if, for each outcome $\omega \in \Omega$, there exists an action $a\in\cA$ such that $\mF_{a}(\omega)\ge c$
\end{definition}

\begin{assumption}\label{ass:smooth}
	There exists $c\ge 0$ such that the delegation problem $\delegation$ is $c$-smooth. 
\end{assumption}

\noindent
This assumption  can be interpreted as requiring that any outcome can be reached with sufficiently high probability, if an adequate amount of effort is devoted to the task. This assumption is natural in general, and it is especially applicable in our motivating example of delegated computation. Our final result will be parametrized in the smoothness parameter. 
On a more technical side, this implies that we do not need to force extremely large payments on single outcomes for incentivizing a particular action. Formally, we can prove the following:

\begin{restatable}{lemma}{lemmacsmooth}\label{lem:csmooth}
	For any $c$-smooth instance, any $a\in \cA$, and any $q\in Q_a$, there exists a $\vp\in\Pi_a$ with $\mF_a^\top\vp=q$ and $\|\vp\|_\infty\le (1+q)/c$.
\end{restatable}
\noindent
As we will see in the following, bounded payments are a necessary condition to bound the variation in expected payment between two actions that induced similar distributions, and hence to exploit \cref{ass:Lip}.

Now, we can use \Cref{th:deltaIC} to shift our focus from solving Problem \ref{prog:pricingcnt} with continuous actions to solving a relaxed problem in which we can allow slight violations of the IC and IR constraints of the user.
In order to do this, for any $\delta>0$, let $\cA_\delta\defeq\{0,\delta,\ldots,1-\delta,1\}$ be the discretization of the continuous action set $\cA=[0,1]$.
Moreover, we define the sets $ Q^{\epsilon}_a(\cA_\delta)$, $a \in \cA_\delta$, as the set of expected payments that can $\epsilon$-induce action $a$ when considering only actions in $\cA_\delta$. Formally, 
\[
Q^\epsilon_a (\cA_\delta)\defeq\left\{q=\mF_a^\top\vp:\, \mF_a^\top\vp-c_a\ge \max\limits_{a'\in\cA_\delta}\mF_{a'}^\top\vp-c_{a'}-\epsilon\right\}.
\]
Then, we define the following discretized and relaxed program parametrized by the discretization  $\delta$:

\begin{equation}\label{prog:pricingcntdiscrete}
	\mleft\{\begin{array}{cll}
	\max\limits_{\{(a_\theta, q_\theta)\}_{\theta\in\Theta}} &\sum\limits_{\theta\in\Theta} \xi_\theta\left(q_{\theta}-c_{a_{\theta}}\right) &\textnormal{s.t.}\\
	&\overline{\mR}^\theta_{a_\theta}-q_{\theta}\ge \underline{\mR}^\theta_{a_{\theta'}}-q_{\theta'}&\forall\theta,\theta'\in\Theta \\
	&\overline{\mR}^\theta_{a_\theta}-q_{\theta}\ge0&\forall\theta,\in\Theta \\
	& q_{\theta}\in Q^{\delta(1+\frac{2}{c})}_{a_{\theta}}(\cA_\delta)&\forall \theta\in\Theta\\
	& a_\theta\in\cA_\delta&\forall\theta\in\Theta
	\end{array}\mright.
\end{equation}
where we defined $\overline{\mR}^\theta_a\defeq\mF_a^\top\mR_\theta+\delta $ and $\underline{\mR}^\theta_a\defeq\mF_a^\top\mR_\theta-\delta $ for every $a\in\cA_\delta$.

Clearly, a feasible solution of Program \ref{prog:pricingcntdiscrete} is $O(\delta)$-IC and $O(\delta)$-IR for the user.
Also, notice that Program~\ref{prog:pricingcntdiscrete} can be solved efficiently whenever $|\cA_\delta|=O(1/\delta)$ is constant, as it fits the same formulation of the problem solved in \Cref{sec:pricing}, and in particular \Cref{cor:constantactions}.
More precisely, we can solve the program in polynomial time whenever $\delta$ is constant.

We prove two useful statements about $\overline{\mR}^\theta_a$, $\underline{\mR}^\theta_a$, and $Q_a^\epsilon$.
First, $\overline{\mR}^\theta_a$ and $\underline{\mR}^\theta_a$ are upper and lower bounds, respectively, of the reward of playing an action ``close'' to $a$. Formally:

\begin{restatable}{lemma}{lemmaCB}\label{lm:CB}
	For all $a\in\cA_\delta$, the following holds:
	\[
	\mF_{a'}^\top\mR_\theta\le\overline\mR_a^\theta\quad\forall a'\textnormal{ s.t. }|a-a'|\le\delta,
	\]
	and similarly:
	\[
	\mF_{a'}^\top\mR_\theta\ge\underline\mR_a^\theta\quad\forall a'\textnormal{ s.t. }|a-a'|\le\delta.
	\]
\end{restatable}

Moreover, we prove that all payments that incentivize an action $a$, also $O(\epsilon/c)$-incentivize all actions $a'$ which are $\epsilon$ away from $a$. This result crucially depends on \Cref{ass:smooth} through \Cref{lem:csmooth}. Indeed, if that assumption is dropped, even small perturbations of the outcome distribution could have enormous consequences on the expected payment required to incentivize certain actions. Formally, we can prove the following:
\begin{restatable}{lemma}{lemmaQChangeA}\label{lem:QChangeA}
	For any $a \in \cA$, $a' \in \cA_\delta$ such that $|a-a'|\le\delta$, it holds that \[Q_a \cap [0,1] \subseteq Q_{a'}^{\delta(1+\frac{2}{c})}(\cA_\delta)  \cap [0,1]  \subseteq Q^{2\delta (1+2/c)}_{a'}  \cap [0,1] .\]
\end{restatable}

Now, we show that relaxing the constraints in Program~\ref{prog:pricingcntdiscrete} has the desired advantages. Indeed, we can prove that a solution to the relaxed program yields larger utility than the payment scheme obtained by solving Program \ref{prog:pricingcnt}. This result follows by \cref{lm:CB}, and \cref{lem:QChangeA}.

\begin{restatable}{lemma}{lemmaRealxCNT}\label{th:relax}
	Let $\PS^\star\defeq\{(a^\star_\theta,q^\star_\theta)\}_{\theta\in\Theta}$ be an optimal solution to Program~\ref{prog:pricingcnt}, and $\tPS:=\{(\tilde a_\theta,\tilde q_\theta)\}_{\theta\in\Theta}$ be an optimal solution to Program \ref{prog:pricingcntdiscrete}. Then:
	\[
	V(\tPS)\ge V(\PS^\star)-\delta,
	\]
	where, for any menu $\PS:=\{(a_\theta,q_\theta)\}_{\theta\in\Theta}$, we defined $V(\PS)\coloneqq\sum_{\theta\in\Theta}\xi_\theta\left(q_\theta-c_{a_\theta}\right)$.
\end{restatable}

Then, we prove the main theorem of this section, which can be easily derived by combining the results discussed above.

\begin{theorem}
	Consider the Program~\ref{prog:pricingcnt} and let $\OPT$ be its optimum. Then, for any constant $\delta>0$, there exists an polynomial-time algorithm which finds an indirect menu that is $O(\sqrt{\delta}/c)$-IC for the service provider and obtains utility at least $\OPT-O(\sqrt{\delta})$.
\end{theorem}

\begin{proof}
	Given a $\delta\in[0,1]$, the algorithm finds a solution $\tPS:=\{(\tilde a_\theta,\tilde q_\theta)\}_{\theta\in\Theta}$ to Program~\ref{prog:pricingcntdiscrete}. Then, $\tilde\PS$ is $O(\delta)$-IC and $O(\delta)$-IR for the user. 
	Moreover, by the IR constraint of the user we have that $\tilde q_\theta\le 1$ for each $\theta \in \Theta$. Hence, by \cref{lem:QChangeA} we have that the menu is $2\delta (1+2/c)$-IC for the service provider.
	Define as $\widetilde{\OPT}$ the utility of such menu. From \Cref{th:relax}, we have that $\widetilde{\OPT}\ge\OPT-\delta$.
	Then, we can apply \Cref{th:deltaIC} to find an indirect menu which is $O(\sqrt{\delta} + \delta/c)$-IC for the service provider, and provides an expected utility of at least $\widetilde{\OPT}\ge\OPT-O(\sqrt{\delta})$.
\end{proof}

%% file: src/random.tex
\section{Menus of Randomized Payment Schemes}\label{sec:randomized}

In this section, we show how to use randomization, employing menus of randomized payment schemes, to overcome the computational hardness results for deterministic menus. 
Moreover, we show that through randomization the service provider can gain a larger expected utility compared to what can be achieved with deterministic menus. This further motivates the use of randomized mechanisms.

First, we introduce the class of \emph{direct menus of randomized payment schemes}, and the corresponding extended interaction between service provide and user.

The interaction between the user and the service provider proceeds as follows:
\begin{enumerate}[label={\roman*)}] 
	\item The service provider commits to a menu of randomized payment schemes \[\Phi\defeq\mleft\{\mleft( \phi^{\theta},(\vp_{\theta,a})_{a \in \cA}\mright)\mright\}_{\theta\in \Theta},\] where $\phi^\theta\in \Delta(\cA \cup \varnothing)$ and $\vp_{\theta,a}\in \mathbb{R}_{\ge0}^m$;
	\item The user's type $\theta\in\Theta$ is drawn accordingly to the type distribution $\vxi$. The type $\theta$ is hidden from the service provider; 
	\item The user selects a randomized payment scheme $(\phi^{\theta'},(\vp_{\theta',a})_{a \in \cA}) \in \Phi$;
	\item The service provider samples an action $a \sim \phi^{\theta'}$;
	\item If the sampled action is $a=\varnothing$ the interaction ends;
	\item Otherwise, the service provider performs action $a\in \cA$ and incurs in cost $c_{a}$. The action $a$ is hidden from the user;
	\item An outcome $\omega$ is sampled accordingly to $\mF_{a}$;
	\item The service provider receives payment of $p_{\theta,a}(\omega)$ from the user.
\end{enumerate}

By following the above interaction, the expected utility for a user of type $\theta \in \Theta$ selecting $(\phi^{\theta'},(\vp_{\theta',a})_{a \in \cA})\in \Phi$ can be written as $\sum_{a \cA}\phi^{\theta'}_a\left[\mF_{a}^\top (\mR_\theta-\vp_{\theta',a})\right]$, while the corresponding expected utility of the service provider is $\sum_{a \cA}\phi^{\theta'}_a\left[\mF_{a}^\top\vp_{\theta',a}-c_{a}\right]$, where by definition we set $\mF_{\varnothing}=\vp_{\theta',\varnothing}=\textbf{0}$, and $c_{\varnothing}=0$.

The optimal direct menu of randomized payment schemes is the solution to the following program.
\begin{subequations}\label{pr:quadratic}
\begin{align}
	\max_{\Phi=\{( \phi^{\theta},(\vp_{\theta,a})_{a \in A})\}_{\theta\in \Theta}} &\sum_{\theta\in\Theta}\xi_\theta \sum_{a \in A} \phi^\theta_a\left(\mF_{a}^\top \vp_{\theta,a}-c_{a}\right)\notag{}\\
	\text{s.t.}\hspace{1cm} & \phi_a^\theta(\mF_{a}^\top \vp_{\theta,a}-c_{a}) \ge  \phi_a^\theta(\mF_{a'}^\top \vp_{\theta,a}-c_{a'})&\forall a,a'\in\cA, \forall \theta\in\Theta\label{eq:ICsp2R}\\
	\phantom{s.t.}\hspace{1cm}& \sum_{a \in A} \phi^{\theta}_a \mF_{a}^\top (\mR_\theta-\vp_{\theta,a}) \ge  \sum_{a \in A} \phi^{\theta'}_a \mF_{a}^\top (\mR_{\theta}-\vp_{\theta',a})&\forall \theta,\theta'\in\Theta\label{eq:ICuserR}\\
	\phantom{s.t.}\hspace{1cm}& \sum_{a \in A} \phi^{\theta}_a \mF_{a}^\top (\mR_\theta-\vp_{\theta,a}) \ge 0&\forall\theta\in\Theta\label{eq:IRR}
\end{align}
\end{subequations}

\noindent The service provider's IC constraints~\eqref{eq:ICsp2R} guarantee that, for each type $\theta\in\Theta$, the service provider is never incentivized to perform an action different from the one that they draw from $\phi^\theta$.
The user's IC constraints~\eqref{eq:ICuserR} guarantee that a user of type $\theta$ is never incentivized to choose a randomized contract different from the one proposed to them (\ie corresponding to type $\theta'$, $\theta'\neq \theta$). Finally, Constraints~\eqref{eq:IRR} guarantee that the user's utility is non-negative. Hence, the user prefers their payment scheme with respect to the opt-out option.

Now, we demonstrate that, in addition to the computational advantages which will be proved later (\Cref{thm:rand poly}), menus of randomized payment schemes yield significantly greater utility to service provider compared to deterministic alternatives.

\begin{restatable}{proposition}{propositiondeterministicBad}\label{prop:deterministicBad}
	For any $n \in \mathbb{N}$, there exists an instance with $n$ user's types in which the expected utility of the service provider by using a menu of randomized payments schemes is at least a multiplicative factor of $\Omega(n)$ larger than the the utility on any deterministic menu of payment schemes.
\end{restatable}

On top of their better guarantees in terms of service provider's expected utility, we demonstrate that randomized menus are also more attractive from a computational standpoint.
More precisely, we show that the optimal menu of randomized payment schemes can be computed in polynomial time.

As we showed above, this problem can be written as a quadratic program (Program \ref{pr:quadratic}).
Here, we  relax it to a linear program by adding variables $x_{\theta,a,\omega}= \phi^\theta_a  \ p_{\theta,a}(\omega)$ for each $\theta \in \Theta,a \in \cA,\omega \in \Omega$, and by replacing this constraint with $x_{\theta,a,\omega}\ge 0$.

\begin{subequations}\label{lp:relax}
\begin{align}
	\max_{\vx\ge \vzero,\vphi \in\Delta^{|\Theta|}_{\cA\cup \varnothing}} &\sum_{\theta\in\Theta} \sum_{a\in\cA}\sum_{\omega\in\Omega}  \xi_\theta F_{a}(\omega) (x_{\theta,a,\omega}- \phi^{\theta}_a c_{a})&&\textnormal{s.t.} \label{eq:relaxObj}\\
	&\sum_{\omega\in\Omega} F_{a}(\omega) (x_{\theta,a,\omega}- \phi^{\theta}_a c_{a})\ge \sum_{\omega\in\Omega}  F_{a'}(\omega) ( x_{\theta,a,\omega}- \phi^{\theta}_a c_{a'})&& \forall \Theta \in \Theta,a, a'\in A\label{eq:relax2}\\
	&\sum_{\omega\in\Omega}\sum_{a\in\cA}  F_{a}(\omega) ( \phi^{\theta}_a R_\theta(\omega)- x_{\theta,a,\omega})\notag{}\\
	& \hspace{2.5cm} \ge \sum_{\omega\in\Omega}\sum_{a\in\cA} F_a(\omega) ( \phi^{\theta'}_a R_\theta(\omega)- x_{\theta',a,\omega})&& \forall \theta,\theta' \in \Theta \label{eq:relax1}\\
	&\sum_{\omega\in\Omega}\sum_{a\in\cA} F_{a}(\omega) (\phi^\theta_a R_\theta(\omega)- x_{\theta,a,\omega})\ge 0 && \forall \theta\in \Theta\label{eq:relax3}
\end{align}
\end{subequations}

\noindent This is a linear program with polynomial size, and hence can be solved in polynomial time.
Then, we show that this is actually a relaxation of the original quadratic program.

 \begin{restatable}{lemma}{lemmarandomrelax}
 	Program~\ref{lp:relax} has at least the same value of Program~\ref{pr:quadratic}.
 \end{restatable}

Finally, we would like to show that, given an optimal solution to~LP \ref{lp:relax}, we can recover a solution to Program~\ref{pr:quadratic} with the same value. This ``recovery step'' is a non-trivial direction, and in many related settings it cannot be carried out. For instance, in the classical contract design framework the optimal menu of randomized contracts may not exist in general, and one has to approximate this recovery step~\citep{gan2022optimal,castiglioni23design,castiglioni23multi}.

As we will see in the following, the main technical challenge is the existence of \emph{irregular} optimal solutions with $\phi^\theta_a=0$ and $x_{\theta,a,\omega} \neq 0$. In absence of such solutions, one could simply invert the variables, and define $p_{\theta, a}(\omega)=x_{\theta,a,\omega}/\phi_a^\theta$ to recover a solution to Program~\ref{pr:quadratic}.

\begin{definition}\label{def:regular}
	A feasible solution to LP~\ref{lp:relax} is said to be \emph{irregular} if there exists a type $\theta \in \Theta$, an action $a \in \cA$, and an outcome $\omega \in \Omega$ such that $\phi^{\theta}_a=0$ and $x_{\theta,a,\omega}\neq 0$.
	A feasible solution to LP~\ref{lp:relax}  is said to be \emph{regular} if it is not irregular.
\end{definition}

Intuitively, given an irregular solution it is not possible to recover a feasible solution to Program~\ref{pr:quadratic} since we should find a $p^{\theta,a}_\omega$ such that  $\phi^{\theta}_a  p^{\theta,a}_\omega= x_{\theta,a,\omega}$.
In classical contract-design problems, the outcome induced by an action depends on the hidden state, and it might be profitable to set a ‘‘payment'' $x_{\theta,a,\omega}$ on an ‘‘not-induced'' action, \emph{i.e.}, $\phi^\theta_a=0$. This would allow the principal to pay the agent's types that are able to induce outcome $\omega$ with high probability. However, this strategy is not effective in our setting since the service provider chooses the action, and the induced outcome does not depend on the type of the user. This allows us to modify an irregular solution into a regular one, without any loss in terms of service provider's utility. This implies that there always exists an optimal solution that is regular. 

\begin{restatable}{lemma}{regularlemma}\label{lm:regular}
Let $(\vx,\vphi)$ be a feasible solution  to Program~\ref{lp:relax}. Then, it is possible to recover in polynomial time a regular solution with at least the same value. 
\end{restatable}

Finally, given a regular solution to Program~\ref{lp:relax}, we can recover a menu of randomized payment schemes, \emph{i.e.}, a solution to Program~\ref{pr:quadratic}, by setting $p_{\theta,a,\omega}= x_{\theta,a,\omega}/\phi^\theta_a$ when $\phi^\theta_a\neq0$, and $p_{\theta,a,\omega}=0$ otherwise for all $\theta\in\Theta,a\in \cA,\omega\in\Omega$.

\begin{restatable}{lemma}{regularLemmaDue}\label{lm:recover}
	Given a regular solution to Program~\ref{lp:relax}, it is possible to recover in polynomial time a solution to Program~\ref{pr:quadratic} with the same value.
\end{restatable}

We conclude that the problem of designing optimal menus of randomized payment schemes can be solved in polynomial time solving Program~\ref{lp:relax}, recovering a regular solution trough \Cref{lm:regular}, and building a menu of randomized payment schemes with the same value through \Cref{lm:recover}.

 \begin{theorem}\label{thm:rand poly}
	There exists an algorithm that solves Program~\ref{pr:quadratic} in polynomial time, thereby yielding an optimal menu of randomized payment schemes.
 \end{theorem}

%% file: src/appendix.tex
\section{Discussion on the service provider Individual Rationality}\label{app:IR}

In this work, we define the problem of computing an optimal menu of (deterministic or randomized) payment schemes without adding the service provider's Individually rationally constraint. This is in contrast with the classical formulation of the principal-agent problem in which the IR constraint must be directly enforced. However, in our problem, the service provider (\ie the agent) designs the payment scheme, and hence the IR constraint is enforced by the fact that the service provider is maximizing their own utility. Formally, the optimal menu of deterministic payments is the solution of Program~\ref{prog:determ} in which we are maximizing the service provider's expected utility. Hence, adding the additional IR cannot increase their utility. Formally, 

\begin{lemma}
	There always exists an optimal menu of deterministic payment schemes that is IR for the service provider.
\end{lemma}
Notice that there might be optimal menus, \emph{i.e.}, optimal solutions to Program~\ref{prog:determ}, in which the service provider's IR constraint is not satisfied. However, this happens only for contracts that are not chosen by any user's type.

\section{Proof of Theorem~\ref{thm:negative}}

	We reduce from \IND, which, given an undirected graph $(V,E)$, is the problem of finding the largest set of vertices $V^*\subseteq V$ such that, for all $v,w\in V^*$, $(v,w)\notin E$ (w.l.o.g.~we can assume that vertices are labeled by their indices so that, for $M\coloneqq|V|$, $V=\range{M}$). In the following, for each undirected graph $(V,E)$, we will build an instance of a delegation problem such that:
	\begin{itemize}
		\item \emph{Soundness:} If there is an independent set of size greater than $M^{2/3}$, then there exists a menu of deterministic payment schemes such that the {service provider}'s utility is at least $\beta M^{5/3}$;
		\item \emph{Completeness:} If all independent sets are of size at most $M^{1/3}$, then the {service provider}'s utility is at most $2\beta M^{4/3}$;
	\end{itemize}
	where $\beta>0$ will be defined in the following. 
	
	\noindent Distinguishing between instances of \IND in which there exists an independent set of size $M^{2/3}$ and instances in which all independent sets are of size at most $M^{1/3}$ is known to be \NPHARD \cite{hastad1999clique,Zuckerman2007linear}. Hence, the construction described above implies that a polynomial time algorithm with a constant multiplicative-factor approximation for the delegation problem does not exist, unless $\mathsf{P}=\mathsf{NP}$.
	
	\xhdr{Construction:} Given a graph $(V,E)$ we build an instance of the delegation problem as follows:
	\begin{itemize}[label=$\diamond$]
		\item Let $N\in\N$ be a parameter that will be set later in the proof;
		\item The set of actions $\cA$ contains one action for each pair $(v,i)\in V\times \range{N}$. The action corresponding to pair $(v,i)$ is denoted by $a_{v,i}$;
		
		\item The set of outcomes $\Omega$ contains one outcome $\omega_{v,i}$ for each pair $(v,i)\in V\times \range{N}$;
		
		\item The set of types $\Theta$ contains one type $\theta_{v,i}$ for each pair $(v,i)\in V\times \range{N}$;
		
		\item Outcome-distributions are such that, for any $a_{v,i}\in\cA$, $\mF_{a_{v,i}}=\delta(\omega_{v,i})$, meaning that action $a_{v,i}$ deterministically induces outcome $\omega_{v,i}$;~\footnote{We denote by $\delta(\cdot)$ the Dirac $\delta$ function.}
		
		\item Costs are zero for every action: $c_a=0$ for all $a\in \cA$;
		
		\item The reward of a user with type $\theta_{v,i}\in\Theta$ for an outcome $\omega_{w,j}\in\Omega$, which we denote by $R_{\theta_{v,i}}(\omega_{w,j})$ (\ie the entry in position $(\theta_{v,i}, \omega_{w,j})$ of $\mR$)  is defined as
		\begin{align}\label{eq:red_utility}
			R_{\theta_{v,i}}(\omega_{w,j}) \coloneqq \begin{cases}
				M^{-Mv-i}&\text{if}\quad w=v \textnormal{ and } i=j\\
				M^{-Mv-i}&\text{if}\quad (v,w)\in E, v<w, \forall j\in\range{N}\\
				0&\text{otherwise}
			\end{cases}
		\end{align}
		\item The type distribution $\vxi$ is such that, for any type $\theta_{v,i}\in\Theta$, $\xi_{\theta_{v,i}} \defeq\beta M^{Mv+i}$ where $\beta^{-1}\defeq\sum_{v\in V, i\in\range{N}}M^{Mv+i}$.
	\end{itemize}
	
	\noindent
	The instance is built in such a way that a menu can extract large payments from types $\{\theta_{v,i}\}_{i\in\range{N}}$ corresponding to node $v$ only by asking for payments on outcomes $\{\omega_{v,i}\}_{i\in\range{N}}$ corresponding to the same node $v$. Moreover, if some payments are allocated to outcomes in $\{\omega_{w,i}\}_{i\in\range{N}}$, relative to a neighbouring node $w$, then the utility that the service provider can extract from the types $\{\theta_{v,i}\}_{i\in\range{N}}$ is small. 
	Intuitively, this is because all the types $\{\theta_{v,i}\}_{i\in\range{N}}$ would gain reward $M^{-Mv-i}$ on every outcome $\omega_{w,j}$ with $w<v$ and $(v,w)\in E$ (see \Cref{eq:red_utility}) and, therefore, they are indifferent among all the contracts that induce one of these outcomes. Hence, all the types $\{\theta_{v,i}\}_{i\in\range{N}}$  would be willing to pay the same amount, corresponding to the contract with minimum payment.
	However, a ``good'' menu  would try to extract a different payment (equal to all the user's reward) from each type in $\{\theta_{v,i}\}_{i\in\range{N}}$, forcing the {service provider} to induce outcomes $\{\omega_{v,i}\}_{i\in\range{N}}$ (see \Cref{claim:red2}).
	Moreover, if an action $a_{v,i}$ is played in a payment scheme proposed to type $\theta_{v,i}$ (as required by \Cref{claim:red2} to gain large utility), then the {service provider} can only extract small payment from all the types $\{\theta_{w,i}\}_{i\in\range{N}}$ relative to neighbouring nodes $w$ s.t. $(w,v)\in E$ (see \Cref{claim:red1}). This follows  from a similar argument to the one that shows that it is possible to extract large payment from types $\{\theta_{v,i}\}_{i\in\range{N}}$ only by placing payments on outcomes $\{\omega_{v,i}\}_{i\in\range{N}}$ corresponding to the same node $v$.

	\xhdr{Soundness:} We show that if there exists an independent set $V^*$ of size at least $M^{2/3}$, then there exists a direct menu of deterministic payment schemes that yields service provider's utility of at least $\beta N M^{2/3}$.
	The menu $\PS$ is defined as follows. For each $v\in V^*$ and $i\in\range{N}$, we let $a_{\theta_{v,i}}=a_{v,i}$ and the vector of payments $\vp_{\theta_{v,i}}$ is defined as:
	\begin{align*}%
		p_{\theta_{v,i}}(\omega)=\begin{cases}
			M^{-Mv-i}&\text{if}\quad\omega=\omega_{v,i}\\
			0&\text{otherwise}
		\end{cases}
	\end{align*}
	In the case in which $v\notin V^*$, we define payment schemes $(a_{\theta_{v,i}}, \vp_{\theta_{v,i}})$ as the payment scheme that gives the highest utility to the user among the payment schemes we just defined for the case of vertices in $V^\ast$ and the opt-out option $(a_\varnothing,\vp_\varnothing)$. In particular, by letting $\Theta^\ast\defeq \{\theta_{v,i}\in\Theta: v\in V^\ast,i\in\range{N}\}$, we have
	\begin{align}\label{eq:red_payments2}
		(a_{\theta_{v,i}}, \vp_{\theta_{v,i}})\in\argmax_{(a,\vp)\in\left\{\mleft(a_{\theta},\vp_{\theta}\mright)\,:\, \theta\in\Theta^\ast\right\}\cup\{(a_\varnothing,\vp_\varnothing)\}} \left\{\mF_{a}^\top\left(\mR_{\theta_{v,i}}-\vp\right)\right\},
	\end{align}
	where the pair $(a_\varnothing, \vp_\varnothing)$ corresponds to the opt-out option and $\vp_\varnothing\defeq \vzero$.
	
	\noindent First, we show that the incentive compatibility and individual rationality constraints are satisfied for the menu $\PS$. 
	For any $v\in V^*$ and $i\in\range{N}$, we have that $p_{\theta_{v,i}}(\omega_{v,i})=M^{-Mv-i}$. Therefore, the IC constraints for the {service provider}
	\[
	p_{\theta_{v,i}}(\omega_{v,i})\ge p_{\theta_{v,i}}(\omega_{w,j})\quad\forall w\in V, j\in\range{M}
	\]
	are satisfied as the right-hand side is either $M^{-Mv-i}$ or $0$. Moreover, for any $v\notin V^*$ and $i\in\range{N}$, it holds $(a_{\theta_{v,i}}, p_{\theta_{v,i}}) \in \left\{\mleft(a_{\theta},\vp_{\theta}\mright)\,:\, \theta\in\Theta^\ast\right\} \cup\{(a_\varnothing,\vp_\varnothing)\}$, and hence the IC constraints for the {service provider} are satisfied also in this case.
	
	\noindent
	Regarding the IC and IR constraints of the user, we observe that for all $v\in V^*, i\in\range{N}$, we have that $R_{\theta_{v,i}}(\omega_{v,i})-p_{\theta_{v,i}}(\omega_{v,i})=0$ by construction. Hence, the user's IR constraint is satisfied.
	Moreover, the user's IC constraints are
	\[
	R_{\theta_{v,i}}(\omega_{v,i})-p_{\theta_{v,i}}(\omega_{v,i})\ge 	R_{\theta_{v,i}}(\omega_{w,j})-p_{\theta_{w,j}}(\omega_{w,j})\quad\forall w\in V^*, j\in\range{N},
	\]
	which are verified since the left-hand side is $0$ and the right-hand side is always smaller then zero, since $p_{\theta_{w,j}}(\omega_{w,j})=M^{-Mw-j}$ for all $w\in V^*, j\in\range{N}$ and $R_{\theta_{v,i}}(\omega_{w,j})$ is $M^{-Mv-i}$ only if $v=w$ and $j=i$ (notice that, by the definition of $V^*$, it holds that $(v,w)\notin E$). 
	If $v\not\in V^*$ then \Cref{eq:red_payments2} ensures that the IC and IR constraints for the user are verified.
	
	\noindent
	Having verified that the menu $\{(a_{\theta_{v,i}}, \vp_{\theta_{v,i}})\}_{v,i}$ satisfies the constraints, we can easily show that such menu provides ``large'' utility to the {service provider}:
	\[
	\sum\limits_{v\in V, i \in\range{N}}\xi_{\theta_{v,i}} p_{\theta_{v,i}(\omega_{v,i})}\ge\sum\limits_{v\in V^*, i \in\range{N}}\xi_{\theta_{v,i}} p_{\theta_{v,i}(\omega_{v,i})}
	=\,\beta|V^*|N=\beta M^{5/3},
	\]
	where we set $N=M$.
	This proves the soundness.

	\xhdr{Completeness: } Now, we show that if all independent sets of $(V,E)$ are of size at most $M^{1/3}$, then all menus of deterministic payment schemes yield an utility of at most $2\beta M^{4/3}$ to the {service provider}. To achieve this, we exploit a crucial aspect of our instance, namely, that the {service provider} cannot extract a large utility from two set of types $\{\theta_{v,i}\}_{i\in\range{N}}$ and $\{\theta_{w,i}\}_{i\in\range{N}}$ which correspond to neighbouring vertices $v$ and $w$. 
	In order to show this, we start by proving that, for any direct menu staisfying IC and IR constraints, when the {service provider} gains large utility from types $\{\theta_{v,i}\}_{i\in\range{N}}$ for a particular vertex $v$, there exist indices $i$ and $j$ such that a user of type $\theta_{v,i}$ is recommended to choose an action $a_{v,j}$ corresponding to the same node $v$.
	\begin{claim}\label{claim:red2}
		Let $\{(a_\theta,\vp_\theta)\}_{\theta\in\Theta}$ be any direct menu that satisfies IC and IR constraint, and let $v\in V$ be such that $\sum_{i\in\range{N}}\xi_{\theta_{v,i}} \mF_{a_{\theta_{v,i}}}^\top \vp_{\theta_{v,i}}> 2\beta$. Then, there exists an index $q\in\range{N}$ such that $a_{\theta_{v,q}}\in\{a_{v,j}\}_{j\in\range{N}}$. 
	\end{claim}
	
	\begin{proof}
		We partition the types $\Theta_v\coloneqq\{\theta_{v,i}\}_{i\in\range{N}}$ into three sets based on the action associated to each type by the service provider:
		\begin{itemize}
			\item $\Theta_{v,1}\coloneqq\left\{\theta:\, a_{\theta}\in\{a_{v,j}\}_{j\in\range{N}}\right\}$
			\item $\Theta_{v,2}\coloneqq\left\{\theta:\, a_{\theta}\in\{a_{w,j}\}_{j\in\range{N}, w\in V}\text{ and }(v,w)\in E,\, w<v\right\}$
			\item $\Theta_{v,3}\coloneqq\Theta_v\setminus(\Theta_{v,1}\cup\Theta_{v,2})$
		\end{itemize}
		Observe that the statement we aim to prove is equivalent to showing that $\Theta_{v,1}$ is non-empty. Moreover, notice that, by construction of user's utilities as per \Cref{eq:red_utility} and IR constraints, we have that the expected payments that can be extracted from types in $\Theta_{v,3}$ must be zero. Indeed, the utility gained at all of these types is $0$. Thus, the {service provider}'s utility from vertex $v$ comes entirely from types belonging to $\Theta_{v,1}$ and $\Theta_{v,2}$. We show that the total utility coming from types in $\Theta_{v,2}$ is not large enough to guarantee a cumulative utility greater than $2\beta$ and, therefore, some utility must come types belonging to $\Theta_{v,1}$, thereby showing that the set is not empty.
		
		\noindent 
		To show that the total payments coming from types in $\Theta_{v,2}$ are small, take any $\theta,\theta'\in\Theta_{v,2}$. The IC constraint for these two types reads as follows
		\[
		\mF_{a_\theta}^\top\left(\mR_\theta-\vp_\theta\right)\ge \mF_{a_{\theta'}}^\top\left(\mR_\theta-\vp_{\theta'}\right).
		\]
		Since both $a_\theta$ and $a_{\theta'}$ induce an outcome $\omega_{w,j}$ with $w<v$ and $j\in\range{N}$, they yield the same utility for the user, \ie $\mF_{a_\theta}^\top\mR_\theta=\mF_{a_{\theta'}}^\top\mR_\theta$. Thus the IC constraints imply that, for all $\theta,\theta'\in\Theta_{v,2}$, we have
		\[
		\mF_{a_{\theta'}}^\top \vp_{\theta'}\ge\mF_{a_\theta}^\top \vp_{\theta}.
		\]
		This implies that, for all $\theta\in\Theta_2$, payments are the same, \ie $\mF_{a_\theta}^\top \vp_{\theta}=\hat p\ge 0$.
		Now, let $i^*\in\range{N}$ be the largest index such that $\theta_{v,i^*}\in\Theta_{v,2}$. Then, the IR constraints for type $\theta_{v,i^*}$ imply that the expected payment $\hat p$ is upper bounded by $M^{-Mv-i^*}$. Then, 
		\begin{align*}
			\sum\limits_{\theta\in\Theta_{v,2}}\xi_\theta\mF_{a_\theta}^\top\vp_{\theta}&\le\beta \hat p \sum\limits_{i\in\range{N}:\,\theta_{v,i}\in\Theta_{v,2}}M^{Mv+i}\\
			&\le\beta M^{-Mv-i^*}\sum\limits_{i\in\range{N}:\, \theta_{v,i}\in\Theta_{v,2}}M^{Mv+i}\\
			&\le 2\beta,
		\end{align*}
		where the first inequality holds by definition of $\vxi$, and the third is obtained by setting $M=N$. This shows that $\Theta_{v,1}$ is non-empty.
	\end{proof}
	
	\noindent The second claim states that, if two vertices are connected, \ie $(v,w)\in E$, and a type associated to vertex $v$ is recommended an action associated to the same node, \ie $a_{\theta_{v,q}}=a_{v,j}$ for some $q,j\in\range{N}$, then the {service provider}'s utility for types corresponding to the neighbouring node $w$ is at most $\beta$. 
	\begin{claim}\label{claim:red1}
		Consider two vertices $v$ and $w$ such that $(v,w)\in E$ with $w<v$, and assume there exists $q \in\range{N}$ such that, for some $j\in\range{N}$, it holds $a_{\theta_{v,q}}=a_{v,j}$. Then,
		\[
		\sum\limits_{i\in\range{N}}\xi_{\theta_{w,i}}\mF_{a_{\theta_{w,i}}}^\top\vp_{\theta_{w,i}}\le\beta
		\]
		holds for any direct menu $\{(a_\theta,\vp_\theta)\}_{\theta\in\Theta}$, which satisfies IC and IR constraints.
	\end{claim}
	
	\begin{proof}
		We show that the expected payments received by a user of type $\theta_{v,q}$ is not too large. Formally,
		\[
		\mF_{a_{\theta_{v,q}}}^\top\vp_{\theta_{v,q}}=p_{\theta_{v,q}}(\omega_{v,j})\le \mF_{a_{\theta_{v,q}}}^\top\mR_{\theta_{v,q}}=R_{\theta_{v,q}}(\omega_{v,j})=M^{-Mv-q}.
		\] 
		where the first inequality comes from the IR constraints, and the last equality holds by construction and since $(v,w)\in E$ (see \Cref{eq:red_utility}). Moreover, the IC constraints for all the users of types $\theta_{w,i}$ imply that
		\begin{align*}
			\mF_{a_{\theta_{w,i}}}^\top\left(\mR_{\theta_{w,i}}-\vp_{\theta_{w,i}}\right)&\ge\mF_{a_{\theta_{v,q}}}^\top\left(\mR_{\theta_{w,i}}-\vp_{\theta_{v,q}}\right)\\
			&\ge \mF_{a_{\theta_{v,q}}}^\top\mR_{\theta_{w,i}}-M^{-Mv-q}.
		\end{align*}
		By observing that $\mF_{a_{\theta_{w,i}}}^\top\mR_{\theta_{w,i}}\le \mF_{a_{\theta_{v,q}}}^\top\mR_{\theta_{w,i}}=M^{-Mw-i}$, we can derive a bound on the expected payment of any user of type $\theta_{w,i}$ which depends only on $v$ and $q$, and not on $w$ and $i$:
		\[
		\mF_{a_{\theta_{w,i}}}^\top\vp_{\theta_{w,i}}\le M^{-Mv-q}.
		\]
		The expected payment received by users of types associated with vertex $w$ can be bounded as follows
		\begin{align*}
			\sum\limits_{i\in\range{N}} \xi_{\theta_{w,i}}\mF_{a_{\theta_{w,i}}}^\top\vp_{\theta_{w,i}}&\le \beta \sum\limits_{i\in\range{N}} M^{Mw+i}M^{-Mv-q}\\
			&\le \beta M^{-M-1}\sum\limits_{i\in\range{N}}M^i\\
			&\le \beta M^{-M}\frac{M^N-1}{M-1}\\
			&\le \beta,
		\end{align*}
		where the last inequality holds by choosing $N=M$. This concludes the proof of the claim.
	\end{proof}

	By combining \Cref{claim:red1} and \Cref{claim:red2} we can easily conclude that, for any $(v,w)\in E$, it cannot happen that the {service provider} collects utility larger then $2\beta$ both from $v$ and $w$. Indeed, if $\sum_{i\in\range{N}}\xi_{\theta_{v,i}} \mF_{a_{\theta_{v,i}}}^\top \vp_{\theta_{v,i}}> 2\beta$ (where w.l.o.g.~we assumed that $v>w$), then by \Cref{claim:red2} we have that the assumptions of \Cref{claim:red1} are satisfied and, therefore, $\sum_{i\in\range{N}}\xi_{\theta_{w,i}}\mF_{a_{\theta_{w,i}}}^\top\vp_{\theta_{w,i}}\le\beta$.
	
	\noindent
	Finally, consider any menu $(a_\theta,\vp_\theta)_{\theta\in\Theta}$, and let $V^*$ be the set of vertices from which the {service provider} gains more then $2\beta$. Then, as we showed above, $V^*$ is an independent set. Moreover, for all $v\not\in V^*$, the service provider can gain at most $\beta$. For every vertex $v$ and index $i\in\range{N}$, we can upper bound the maximum payment by using the IR constraint: \[\mF_{a_{\theta_{v,i}}}^\top\vp_{\theta_{v,i}}\le\mF_{a_{\theta_{v,i}}}^\top\mR_{\theta_{v,i}}\le M^{-Mv-i}.\] This leads to
	\begin{align*}
		\sum\limits_{v\in V, i\in\range{N}} \xi_{\theta_{v,i}}\mF_{a_{\theta_{v,i}}}^\top\vp_{\theta_{v,i}}&=\sum\limits_{v\in V^*, i\in\range{N}}\xi_{\theta_{v,i}}\mF_{a_{\theta_{v,i}}}^\top\vp_{\theta_{v,i}}+\sum\limits_{v\not\in V^*, i\in\range{N}}\xi_{\theta_{v,i}}\mF_{a_{\theta_{v,i}}}^\top\vp_{\theta_{v,i}}\\
		&\le \beta |V^*|N+\beta M \\
		&\le 2\beta M^{4/3},
	\end{align*}
	which holds by choosing $N=M$. This proves the completeness and concludes the proof.

\section{Computing optimal menus of small complexity (Section \ref{sec:pricing})}

\shapeofQ*

\begin{proof}
	The set $Q_a$ is the image of $\Pi_a$ under the linear transformation $\vp\mapsto \mF_a^\top \vp$.
	Then, note that $\Pi_a\subseteq \reals^m_{\ge 0}$ is unbounded for all $a\in\cA$. Indeed, it is easy to verify that for all $\vp\in\Pi_a$ then also $\vp+x\vone\in\Pi_a$, where $x\in\reals_{\ge0}$ and $\vone\in\reals^m$ is a vector of all ones. Thus, since the linear transformation $\vp\mapsto \mF_a^\top \vp$ is of positive components, we have that $Q_a$ is a positive ray of the form $[l_a,+\infty)$ for some $l_a\ge0$.
\end{proof}

\lemmapricingtodeterm*
\begin{proof}
	For each $i \in \range{k}$, we consider a payment scheme $(a_i,\vp_i)$ where, as prescribed by the Lemma, action $a_i$ is the same of $(a_i,q_i)$. Moreover, the payment vector $\vp_i$ can be set to any payment vector such that $q_i=\mF_{a_i}^\top\vp_i$. It is always possible to construct a payment vector with these properties. Indeed, $\vp_i$ can be easily computed in polynomial time by solving the linear feasibility problem
	\[
	\mleft\{\vp\in\reals_{\ge 0}^m: \mF_{a_i}^\top\vp=q_{i}\textnormal{ and }\vp_i\in\Pi_{a_i}\mright\}.
	\]
	First, we show that $\PS=\{(a_i, \vp_i)\}_{i\in\range{k}}$ is feasible.
	For each $i \in \range{k} $, it holds $\vp_i \in \Pi_{a_i}$ and, for all $a'\in\cA$, $\mF_{a_i}^\top\vp_i-c_{a_i}\ge \mF_{a'}^\top\vp_i-c_{a'}$, which are exactly the constraints of Program \eqref{prog:determ} relative to payment scheme $i$.
	Moreover, the value of the objective for is the same since it depends only on the expected payments $q_i=\mF_{a_i}^\top\vp_i$, $i \in \range{k}$, and not on the realized payments $\vp_i$, $i \in \range{k}$. 
\end{proof}

\lemmapaymenttopricing*
\begin{proof}
	It is enough to consider, for each $i\in\range{k}$, the pair $(a_i,q_i)=(a_i,\mF_{a_i}^\top\vp_i)$. Since $(a_i,\vp_i)$ is feasible for Problem \eqref{prog:determ}, it holds $q_i\in Q_{a_i}$. By construction of \Cref{eq:selection func}, the first set of constraints of Problem \eqref{prog:pricing} also holds. Finally, the objectives yield identical values for the same argument employed in the previous lemma.
\end{proof}

\polyalgo*
\begin{proof}
	Fix any function $f:\Theta\to\range{k}\cup\{\varnothing\}$ and $\va\in\cA^k$, and consider the linear program $\LP(f,\va)$ parameterized by $(f,\va)$:
	\begin{align}\label{prog:pricing}
		\LP(f,\va)\coloneqq\begin{cases}
			\max\limits_{\vq\in\reals^k} \sum\limits_{\theta\in\Theta} \xi_\theta\left(q_{f(\theta)}-c_{a_{f(\theta)}}\right) &\textnormal{s.t.}\\
			\mF^\top_{a_{f(\theta)}}\mR_\theta-q_{f(\theta)}\ge \mF^\top_{a_{f(\theta')}}\mR_{\theta}-q_{f(\theta')},&\forall\theta,\theta'\in\Theta, \theta'\neq\theta\nonumber\\
			q_{i}\in Q_{a_i},&\forall i\in\range{k}\nonumber
		\end{cases}
	\end{align}
	where $Q_{a_i}$ is defined as in \Cref{def:Q} for all $i\in\range{k}$.
	By standard arguments of linear programming, either the solution of $\LP(f,\va)$ is $+\infty$, or there exists an optimal solution on an extreme point. We claim that the optimal value is bounded. Indeed, take any feasible point $\vq$ and any $\theta$, and consider the point $\tilde\vq\defeq\vq+\alpha \cdot\ve_{f(\theta)}$ for $\alpha\ge 0$. Then, it is easy to see that for $\alpha$ large enough the point $\tilde\vq$ is no longer feasible. Thus, the solution is on a vertex of the polytope and, in particular, it is at the intersection of $k$ hyperplanes characterizing the feasibility set of $\LP(f,\va)$.
	In particular, for a fixed pair $(f,\va)$, an optimal solution is at the intersection of $k$ hyperplanes among $\{H_{\va}^{i,j,\theta}\}_{i,j\in\range{k},\theta\in\Theta}$, and $\{\tH_{\va}^{i}\}_{i\in\range{k}}$. The former set of constraints has cardinality $nk^2$, while the latter has cardinality $k$. We observe that both sets are independent from $f$. 
	Therefore, the possible combinations of these hyperplanes are ${k(kn+1)\choose k}=O\left((kn)^k\right)$ and, therefore, $|\cF_{\va}|=O((kn)^k)$ for each $\va$. 
	There are $\ell^k$ possible choices for $\va\in\cA^k$. For each $\va\in\cA^k$, it is enough to evaluate the value of the objective for all possible $\vq\in\cF_{\va}$, which are $O\left((kn)^k\right)$. We observe that, once the pair $(\va,\vq)$ is fixed, then the selection function $i(\cdot)$ is fully specified and, therefore, we can efficiently evaluate it for each $\theta\in\Theta$. 
	Putting it all together, an optimal solution can be computed in time polynomial in  $(kn)^k \cdot \ell^k$, which concludes the proof.
\end{proof}

\lemmaminactions*
\begin{proof}
	Let $\PS=\{(\adir_\theta,\pdir_\theta):\theta \in \Theta\}$ be an optimal direct menu.
	Since $\ell\le n$, some actions will be chosen for multiple types. Users whose types correspond to actions associated with multiple other types, will select the payment scheme, among those of types sharing the same action, that is cheaper in expectation. Let $\cA^{(\textsc{d})}\defeq \bigcup_{\theta}\{\adir_\theta\}$. For each action $a\in\cA^{(\textsc{d})}$, the ``cheapest'' types associated to $a$ are 
	\[
	\hat \theta(a)\in \argmin_{\theta \in \Theta:\adir_\theta=a} \mF_a^\top\vp_\theta.
	\]
	Let $k=|\cA^{(\textsc{d})}|$. We build an indirect menu of size $k$ as follows: the menu employs $k$ distinct actions which we denote by $\aind_1,\ldots,\aind_k$, one for each of the actions in $\cA^{(\textsc{d})}$. Then, the menu of size $k\le\ell$ is 
	\[
	\bigcup_{i\in\range{k}}\mleft\{\mleft(\aind_i,\pind_i\defeq\pdir_{\theta}\mright)\,:\,\theta\in \hat\theta\mleft(\aind_i\mright)\mright\}.
	\]
	Every action $\aind_i$ is clearly IC, since it is chosen, together with the same payment scheme, within the direct menu $\PS$. Moreover, for each $\theta\in\Theta$,
	\[
	\mF^\top_{\aind_{i(\theta)}}\mleft(\mR_\theta  - \pind_{i(\theta)}\mright)\ge \mF^\top_{\adir_{\theta}}\mleft(\mR_\theta  - \pdir_{\theta}\mright)\ge \mF^\top_{\adir_{\theta'}}\mleft(\mR_\theta  - \pdir_{\theta'}\mright)=\mF^\top_{\aind_{j}}\mleft(\mR_\theta  - \pind_{j}\mright),
	\]
	for every payment scheme $j\in\range{k}$ and $\theta'\in\Theta$ such that $\mleft(\adir_{\theta'},\pdir_{\theta'}\mright)=\mleft(\aind_j,\pind_j\mright)$. Therefore, the indirect menu is IC also for the user. Finally, we have that, for each $\theta\in\Theta$, $\aind_{i(\theta)}=\adir_\theta$. This holds by construction and by the user's IC: a user of type $\theta$ will choose $\theta'\in(\theta)$ so that type $\theta'$ is recommended the same action in the direct menu, and the payment schemes for $\theta$ and $\theta'$ must be such that $\pdir_\theta=\pdir_{\theta'}$ by IC of the direct menu. Therefore, $\OPT_\ell = \OPT$. This concludes the proof.
\end{proof}

\section{Trade off between utility and complexity of menus (Section \ref{sec:classes})}

\theoremsingleGood*

\begin{proof}
	First, we observe that by definition $\OPT=\OPT_n$, and by \Cref{lem:smallactions} we have $\OPT=\OPT_\ell$ whenever $\ell\le n$. Therefore, there always exists an optimal menu of size $\eta\defeq\min(\ell,n)$.
	Consider an optimal menu of payments schemes $\PS=\{(a_i,\vp_i)\}_{i \in \range{\eta}}$. We describe a procedure for building a menu $\tPS=\{(\tilde a_i, \tilde{\vp}_i)\}_{i\in\range{k}}$ of size $k$ that guarantees utility at least $k\cdot\OPT/\eta$. The procedure simply selects the payment schemes among the $n$ direct options that maximize the utility for the {service provider}. Formally, let $\cI \subseteq \range{\eta}$ be a set of size $k=|\cI|$ such that, for any alternative $\cI'\subseteq \range{\eta}$, $|\cI'|=k$, 
	\begin{equation*}\label{eq:highestgrossing}
		\sum_{j\in\cI}\,\sum_{\theta: i(\theta|\PS)=j}\xi_\theta\left(\mF_{ a_{j}}^\top \vp_{j}-c_{j}\right) \ge \sum_{j\in\cI'}\,\sum_{\theta: i(\theta|\PS)=j} \xi_{\theta}  \left(\mF_{a_{j}}^\top\vp_{j}-c_{a_{j}}\right).
	\end{equation*}
	Intuitively, $\cI$ contains the $k$ payment schemes yielding the highest service provider's expected utility.
	Then, let $\tPS=\{(\tilde a_i, \tilde{\vp}_i)\}_{i\in\range{k}}$ be the relabeling of payment schemes $\{( a_i, {\vp}_i)\}_{i\in \cI}$. We show that $\PS$ extracts an expected utility which is big enough:
	\begin{align*}
		&\sum \limits_{\theta\in\Theta}\xi_\theta\left(\mF^\top_{\tilde a_{\ui(\theta|\tPS)}}\tilde \vp_{\ui(\theta|\tPS)}-c_{\tilde a_{\ui(\theta|\tPS)}}\right)\\
		&\hspace{1cm}=\sum\limits_{i\in \cI} \sum_{\theta :   i(\theta|\PS)=i} \xi_\theta\left(\mF^\top_{  a_{i}}\vp_{i}-c_{ a_{i}}\right) + \sum_{\theta: i(\theta|\PS)\notin \cI}     \xi_\theta \left(\mF^\top_{\tilde a_{\ui(\theta|\tPS)}}\tilde \vp_{\ui(\theta|\tPS)}-c_{\tilde a_{\ui(\theta|\tPS)}}\right)\\
		&\hspace{1cm}\ge \frac{k}{\eta} \sum\limits_{i\in \range{\eta}} \sum_{\theta :   i(\theta|\PS)=i} \xi_\theta\left(\mF^\top_{ a_{i}}\vp_{i}-c_{ a_{i}}\right)  + \sum_{\theta: i(\theta|\PS)\notin \cI}     \xi_\theta \left(\mF^\top_{\tilde a_{\ui(\theta|\tPS)}}\tilde \vp_{\ui(\theta|\tPS)}-c_{\tilde a_{\ui(\theta|\tPS)}}\right)\\
		&\hspace{1cm}\ge  \frac{k}{\eta} \sum\limits_{i\in \range{\eta}} \sum_{\theta :   i(\theta|\PS)=i} \xi_\theta\left(\mF^\top_{ a_{i}}\vp_{i}-c_{ a_{i}}\right) \\
		&\hspace{1cm}= \frac{k}{\eta} \sum_{\theta \in \Theta}  \xi_\theta\left(\mF^\top_{ a_{\ui(\theta|\PS)}}\vp_{\ui(\theta|\PS)}-c_{ a_{\ui(\theta|\PS)}}\right) \\
		&\hspace{1cm}=\frac{k}{\eta}  \OPT,
	\end{align*}
	where the first inequality follows by the definition of $\cI$, and the second one by IR (\ie, $\left(\mF^\top_{ a_{i}}\vp_{i}-c_{ a_{i}}\right)\ge 0$ for each $i \in \range{k}$). This concludes the proof.
\end{proof}

\propositionsingleBad*

\begin{proof}
	Consider the following instance $\delegation$: the set of types is $\Theta\defeq\{\theta_1,\ldots,\theta_n\}$, the set of actions is $\cA\defeq\{a_1,\ldots,a_n\}$, and the outcome space is $\Omega\defeq\{\omega_1,\ldots, \omega_n\}$. Each action $a_i$ induces deterministically $\omega_i$, \ie $\mF_{a_i}=\delta(\omega_i)$ for all $i\in\range{n}$. The type distribution is uniform, \ie $\xi_\theta=1/n$ for all $\theta$. The reward of an agent $\theta_i$ is $1$ if the outcome reached is $\omega_i$, \ie $\mR_{\theta_i,\omega}=\indicator{\omega=\omega_i}$, and $\mR_{\theta_i,\omega}=0$ otherwise.
	
	The optimal direct menu $\{(a_\theta, \vp_\theta)\}_{\theta\in\Theta}$ is such that $a_{\theta_i}=a_i$ and $p_{\theta_i,\omega}=\indicator{\omega=\omega_i}$. This yields an expected utility of $1$ for the {service provider}.
	On the other hand, the optimal menu of size $k$ can employ at most $k$ actions. For each type $\theta_i$, if action $a_i$ is not among those proposed in the menu, the user expected utility is $0$, and hence the payment must be $0$ due to the individual rationality constraints.
	For all the other types $\theta_i$ for which action $a_i$ is proposed, the {service provider} utility is at most $1$.
	Therefore, the overall service provider's utility is at most $k/n$ for a menu of size $k$.
\end{proof}

\section{Continuous Actions and Robustness (Section \ref{sec:continuous})}

\lemmaQeps*

\begin{proof}
	Take $q\in Q^{\epsilon}_a$. By the definition of   $Q^{\epsilon}_a$, there exists a payment scheme $\vp$ such that \[ \mF_a^\top\vp-c_a\ge \max\limits_{a'\in\cA}\mF_a^\top\vp-c_{a'}-\epsilon.\]
	Moreover, since $0\le \epsilon'\le q$, it is easy to see that there exists $\vp'\in \mathbb{R}_{\ge 0}^n$ such that $(\vp-\vp')\ge0$ and  $\mF_a^\top\vp'= \epsilon'$. Let $\vp^\star\defeq(\vp-\vp')$.
	Then,
	\[ \mF_a^\top\vp^\star-c_a \ge \mF_a^\top\vp-c_a - \mF_a^\top \vp' =  \mF_a^\top\vp-c_a -\epsilon'\ge \max\limits_{a'\in\cA}\mF_a^\top\vp-c_{a'}-\epsilon-\epsilon'.  \]
	By the definition of $Q^{\epsilon+\epsilon'}_a$ and $\mF_a^\top\vp^\star=q-\epsilon'$, it follows that $q-\epsilon'\in  Q^{\epsilon+\epsilon'}_a$.
\end{proof}

\begin{lemma}\label{lem:qle1}
	For all types $\theta\in\Theta$ and menu $P=\{(a_i,q_i)\}_{i\in\range{k}}$, we have that $q_{i(\theta|P)}\le1$.
\end{lemma}
\begin{proof}
	Since type $\theta\in\Theta$ selects th payment with index $i(\theta|P)$ we have that the IR constraints is satisfied and thus $q_{i(\theta|P)}\le \mF_{a_\theta}^\top\mR_\theta\le 1$.
\end{proof}

\lemmadeltaIC*
\begin{proof}
	We represent $\tPS$ through its pricing equivalent $\{(\tilde a_\theta, \tilde q_\theta)\}_{\theta\in\Theta}$ (this is always possible thanks to \Cref{lem:payment to pricing}). First, we observe that we can always represent the direct menu $\{(\tilde a_\theta, \tilde q_\theta)\}_{\theta\in\Theta}$ as an equivalent indirect menu $ \PS^{\ind}:=\big\{\big( \aind_i, q^{\ind}_i\big)\big\}_{i\in\range{k}}$, where $k=n$ .
	For each $i\in\range{k}$, let $g_i\defeq \alpha ( q^{\ind}_i-c_{a^{\ind}_i})$, where $\alpha$ is a parameter which will be specified in the following. Then, define the menu $\PS:=\{(a_i, q_i)\}_{i\in\range{\tilde k}}$ where $a_i= a^{\ind}_i$ and $q_i= q^{\ind}_i-g_i$ if $q_i^\ind-c_{a_i^\ind}\ge\sqrt{2\delta}$, otherwise it is not included in $\PS$. Thus, $\tilde k\le k$, is the number of items for which $q_i^\ind-c_{a_i^\ind}\ge\sqrt{2\delta}$.
	
	Define $\hat\Theta$ all the types such that $q_{\ui(\theta|\PS^\ind)}^\ind-c_{a_{\ui(\theta|\PS^\ind)}}^\ind\ge\sqrt{2\delta}$.
	
	Now, we prove that for all types $\theta\in\hat\Theta$, after subtracting $g_i$ to all the payments a user of type $\theta$, the types that choose a different payment scheme, would choose a new payment scheme with ``large'' utility for the {service provider}. Formally:
	\begin{claim}\label{claim:ggeg}
		For any $\theta\in\hat\Theta$, it holds that $g_{\ui(\theta|\PS)}\ge g_{\ui(\theta\vert \PS^{\ind})}-\delta$.
	\end{claim}
	\begin{proof}
		By definition of $\ui{(\cdot)}$ we obtain that, for all $\theta$ and $j\in\range{\tilde k}$, the following holds:
		\[
		\mF_{a_{\ui(\theta|\PS)}}^\top\mR_\theta-q_{\ui(\theta|\PS)}\ge\mF_{a_{j}}^\top\mR_\theta-q_{j}.
		\]
		By plugging the definition of $q_i$ on both sides and by setting $j=\ui(\theta\vert\PS^{\ind})$, (which can be done as contract $j=\ui(\theta\vert\PS^{\ind})$ is included in $\tilde \PS$ as $q^\ind_{j}-c_{a_{j}^\ind}\ge\sqrt{2\delta}$ by assumption that $\theta\in\hat\Theta$) we obtain
		\begin{align*}
			\mF_{a_{\ui(\theta|\PS)}}^\top\mR_\theta- q^{\ind}_{\ui(\theta|\PS)}+g_{\ui(\theta|\PS)}&\ge \mF_{ \aind_{\ui(\theta| \PS^{\ind})}}^\top\mR_{\theta}- q^{\ind}_{\ui(\theta| \PS^{\ind})}+g_{\ui(\theta| \PS^{\ind})}\\
			&\ge \mF_{a_{\ui(\theta|\PS)}}^\top\mR_\theta- q^{\ind}_{\ui(\theta|\PS)}+g_{\ui(\theta| \PS^{\ind})} -\delta,
		\end{align*}
		where the last inequality follows from $\PS^{\ind}$ being $\delta$-IC.
		Simplifying $\mF_{a_{\ui(\theta|\PS)}}^\top\mR_\theta- q^{\ind}_{\ui(\theta|\PS)}$ concludes the proof of \Cref{claim:ggeg}.
	\end{proof}
	\noindent
	Now, consider the utility extracted through $\PS$ from a user of type $\theta\in\hat\Theta$.
	
	First, consider the types such that $\ui(\theta|\PS)\neq \ui(\theta| \PS^\ind)$:
	\begin{align}
		q_{\ui(\theta|\PS)}-c_{a_{\ui(\theta|\PS)}}&= q^{\ind}_{\ui(\theta|\PS)}-c_{ a^{\ind}_{\ui(\theta|\PS)}}-g_{\ui(\theta|\PS)}\nonumber \tag{Definition of $q_i$}\\
		&\geq q^{\ind}_{\ui(\theta|\PS)}-c_{ a^{\ind}_{\ui(\theta|\PS)}}-\alpha\nonumber \tag{\Cref{lem:qle1}}\\
		&= \frac{g_{\ui(\theta|\PS)}}{\alpha}-\alpha\nonumber\\
		&\ge \frac{g_{\ui(\theta| \PS^{\ind})}-\delta}{\alpha}-\alpha\nonumber\tag{\Cref{claim:ggeg}}\\
		&= q_{\ui(\theta|  \PS^{\ind})}^\ind-c_{ a^\ind_{\ui(\theta| \PS^{\ind})}}-\frac{\delta}{\alpha}-\alpha.\nonumber
	\end{align}
	Second consider, $\ui(\theta|\PS)=\ui(\theta| \PS^\ind)$, then we immediately have that 
	\begin{align}	
		q_{\ui(\theta|\PS)}-c_{a_{\ui(\theta| \PS)}}\ge q^\ind_{\ui(\theta| \PS^\ind)}-c_{a^\ind_{\ui(\theta| \PS^\ind)}}-\alpha,\nonumber
	\end{align}
	where we used again that $g_{i}\le \alpha$ (\Cref{lem:qle1}).
	
	So, for all types $\theta\in\hat\Theta$ it will happen that:
	\begin{equation}\label{eq:boundhatTheta}
		q_{\ui(\theta|\PS)}-c_{a_{\ui(\theta|\PS)}}\ge q_{\ui(\theta|  \PS^{\ind})}^\ind-c_{ a^\ind_{\ui(\theta| \PS^{\ind})}}-\frac\delta\alpha-\alpha
	\end{equation}
	
	On the other hand, for $\alpha=\sqrt{\delta}$, for all types $\theta\in\Theta\setminus\hat\Theta$, we have that they will choose a payment scheme $(q_i,a_i)\in\PS$ such that:
	\begin{align}
		q_i-c_{a_i}&=q_i^\ind-g_i-c_{a_i^\ind}\tag{Definition of $q_i$}\\
		&\ge \sqrt{2\delta}-g_i \tag{$\theta\in\hat\Theta$}\\
		&\ge \sqrt{2\delta}-\alpha\tag{Definition of $g_i$ and \Cref{lem:qle1}}\\
		&\ge 0\tag{$\alpha=\sqrt{\delta}$},
	\end{align}
	and thus for each of those types we will loose in utility at most $\sqrt{2\delta}$.
	Formally:
	\begin{equation}\label{eq:boundhatTheta2}
		\sum\limits_{\theta\in\Theta\setminus\hat\Theta} \xi_\theta\left[q_{\ui(\theta|\PS)}-c_{a_{\ui(\theta| \PS)}}\right]\ge \sum\limits_{\theta\in\Theta\setminus\hat\Theta}\xi_\theta \left[q^\ind_{\ui(\theta|\PS^\ind)}-c_{a^\ind_{\ui(\theta| \PS^\ind)}}\right] -\sqrt{2\delta}
	\end{equation}
	
	Now, we proceed to bound ${\OPT}$ w.r.t~$\widetilde\OPT$.
	
	\begin{align*}
		\OPT&:=\sum\limits_{\theta\in\Theta} \xi_\theta\left[q_{\ui(\theta|\PS)}-c_{a_{\ui(\theta| \PS)}}\right]\\
		&=\sum\limits_{\theta\in\hat\Theta} \xi_\theta\left[q_{\ui(\theta|\PS)}-c_{a_{\ui(\theta| \PS)}}\right]+\sum\limits_{\theta\in\Theta\setminus\hat\Theta} \xi_\theta\left[q_{\ui(\theta|\PS)}-c_{a_{\ui(\theta| \PS)}}\right]\\
		&\ge \sum\limits_{\theta\in\hat\Theta} \xi_\theta\left[q_{\ui(\theta|\PS)}-c_{a_{\ui(\theta| \PS)}}\right]+\sum\limits_{\theta\in\Theta\setminus\hat\Theta}\xi_\theta \left[q^\ind_{\ui(\theta|\PS^\ind)}-c_{a^\ind_{\ui(\theta| \PS^\ind)}}\right] -\sqrt{2\delta}\tag{\Cref{eq:boundhatTheta2}}\\
		&\ge \sum\limits_{\theta\in\hat\Theta} \xi_\theta\left[q^\ind_{\ui(\theta|\PS^\ind)}-c_{a^\ind_{\ui(\theta| \PS^\ind)}}\right]-2\sqrt{\delta}+\sum\limits_{\theta\in\Theta\setminus\hat\Theta}\xi_\theta \left[q^\ind_{\ui(\theta|\PS^\ind)}-c_{a^\ind_{\ui(\theta| \PS^\ind)}}\right] -\sqrt{2\delta}\tag{\Cref{eq:boundhatTheta} and $\alpha=\sqrt{\delta}$}\\
		&=\widetilde\OPT-2\sqrt{\delta}-\sqrt{2\delta}\\
		&\ge\widetilde{\OPT}-4\sqrt{\delta} 
	\end{align*}

	we obtain that the service provider's  utility for the menu $\PS$ is at least $\widetilde{OPT}- O(\sqrt{\delta})$ for each realized type $\theta$. This proves the first point of the statement. 
	
	The second point of the statement follows by \Cref{lem:Qeps}, by observing that, for each $i\in\range{k}$, $q^{\ind}_i\in Q_{\aind_i}^\epsilon$ by assumption, and that $q_i=q_i^{\ind}-g_i\ge q_i^{\ind}-\alpha\ge 0$, where the last inequality come from the fact that $q_i$ is included only if $q_i^\ind-c_{a_i^\ind}\ge\sqrt{2\delta}$, and thus $q_i^\ind-\alpha\ge \sqrt{2\delta}-\sqrt{\delta}\ge 0$. Therefore, $q_i\in Q_{a_i^{\ind}}^{\epsilon+\alpha}$ and LL constraints are satisfied.
	
\end{proof}

\lemmacsmooth*
\begin{proof}
	Fix any $\omega\in\Omega$. By \Cref{ass:smooth} there exists an action $a_{\omega}$ such that $F_{a_{\omega}}(\omega)\ge c$.
	\begin{align*}
		q&=\mF_a^\top\vp\\
		&\ge  \mF_a^\top\vp-c_{a}\tag{$c_a\ge 0$}\\
		&\ge  \mF_{a_\omega}^\top\vp-c_{a_\omega}\tag{$p\in\Pi_a$}\\
		&=\sum\limits_{\omega'\neq\omega} F_{a_\omega}(\omega')p(\omega')+ F_{a_\omega}(\omega)p(\omega)-c_{a_\omega}\\
		&\ge c p(\omega)-1\tag{$F_{a_\omega}(\omega)\ge c$ and $c_{a_\omega}\le 1$}
	\end{align*}
	rearranging we obtain that for all $\omega\in\Omega$ we have $p(\omega)\le (1+q)/c$ which concludes the proof.
\end{proof}

\lemmaCB*
\begin{proof}
	We are going to prove the first part of the statement as the second can be derives in a similar fashion. The statement easily follows from \Cref{ass:Lip} and Holder inequality:
	\begin{align}
		|\mF_{a'}^\top\mR_\theta- \mF_a^\top\mR_\theta|\le \|\mF_{a'}-\mF_{a}\|_1\cdot\|\mR_\theta\|_\infty\le |a'-a|\cdot\|\mR_\theta\|_\infty\le \delta.
	\end{align}
	which concludes the proof.
\end{proof}

\lemmaQChangeA*
\begin{proof}
	Take $q\in Q_a \cap [0,1]$. By the definition of   $Q_a$, there exists a payment scheme $\vp$ such that \[ \mF_a^\top\vp-c_a\ge \max\limits_{\hat a\in\cA}\mF_{\hat a}^\top\vp-c_{\hat a}.\]
	Hence, by the Lipschitz continuity of the costs and distributions, we have that
	\[
	\mF_{a'}^\top\vp-c_{a'}  \ge \mF_{a}^\top\vp-c_{a} - \|\mF_{a'}-\mF_{a}\|_1\|\vp\|_\infty - \delta \ge \mF_{a}^\top\vp-c_{a} - \|\vp\|_\infty\delta - \delta \ge \max\limits_{\hat a\in\cA}\mF_{\hat a}^\top\vp-c_{\hat a} - \delta - \|\vp\|_\infty\delta.
	\]
	Moreover, $\|\vp\|_\infty\le \frac{1+q}{c}$ by \cref{lem:csmooth}, and $q \in[0, 1]$. Thus, can conclude that
	\[
	\mF_{a'}^\top\vp-c_{a'}\ge \max\limits_{\hat a\in\cA}\mF_{\hat a}^\top\vp-c_{\hat a} - \delta (1+2/c).
	\]
	Hence, 
	\[q \in Q_{a'}^{\delta (1+2/c)}.  \]
	This proves $Q_a \subseteq Q_{a'}^{\delta (1+\frac{2}{c})}(\cA_\delta) \cap [0,1]$.
	
	Then, we prove the second part of the lemma.
	Take an $q \in Q_{a'}^{\delta (1+\frac{2}{c})} \cap [0,1]$.
	By the definition of   $Q_{a'}^{\epsilon (1+\frac{2}{c})}$, there exists a payment scheme $\vp$ such that \[ \mF_{a'}^\top\vp-c_{a'}\ge \max\limits_{\hat a\in\cA_\delta}\mF_{\hat a}^\top\vp-c_{\hat a} - \delta (1+\frac{2}{c}) .\]
	Hence, by the Lipschitz continuity of the costs and distributions, we have that
	\[
	\mF_{a'}^\top\vp-c_{a'}  \ge \max\limits_{\hat a\in\cA_\delta}\mF_{\hat a}^\top\vp-c_{\hat a} - \delta (1+\frac{2}{c}) \ge  \max\limits_{\hat a\in\cA}\mF_{\hat a}^\top\vp-c_{\hat a} - \delta (1+\frac{2}{c}) - \delta \lVert \vp\rVert_{\infty}-\delta ,
	\]
	were the second inequality follows observing that each action $\hat a\in \cA$ has distance at most $\delta$ from an action in $\cA_\delta$. 
	Moreover, $\|\vp\|_\infty\le \frac{1+q}{c}$ by \cref{lem:csmooth}, and $q \in[0, 1]$. Thus, can conclude that
	\[
	\mF_{a'}^\top\vp-c_{a'}\ge \max\limits_{\hat a\in\cA}\mF_{\hat a}^\top\vp-c_{\hat a} - 2\delta (1+\frac{2}{c}).
	\]
		Hence, 
	\[q \in Q_{a'}^{2\delta (1+2/c)}.  \]
	This proves $Q_{a'}^{\delta (1+\frac{2}{c})}(\cA_\delta) \cap [0,1]\subset Q_{a'}^{2\delta (1+2/c)} \cap [0,1] $.
	
\end{proof}

\lemmaRealxCNT*
\begin{proof}
	First, we show how to modify $\PS^\star$ and construct a feasible solution to Program~\ref{prog:pricingcntdiscrete}, without changing the utility too much. Consider the menu $\hat\PS^\star:=\{(\hat a^\star_\theta, \hat q^\star_\theta)\}_\theta$ where $\hat a^\star_\theta$ is the projection of $a^\star_\theta$ onto $\cA_\delta$, \ie the closest element to $\hat a^\star_\theta$ in $\cA_\delta$,  and $\hat q^\star_\theta=q^\star_\theta$.
	Then it is easy to see that:
	\[
	V(\hat\PS^\star)\ge V(\PS^\star)-\delta,
	\]
	by the Lipschitz continuity of the costs. 
	Now, we prove that $\hat\PS^\star:=\{(\hat a^\star_\theta, \hat q^\star_\theta)\}_\theta$ is feasible for Program~\ref{prog:pricingcntdiscrete}. 
	
	Take a $\theta \in \Theta$. First, by construction $\hat a^\star_\theta\in\cA_\delta$. Moreover, for any $\theta' \in \Theta$:
	\[
	\overline{\mR}^\theta_{\hat a^\star_\theta}-\hat q^\star_{\theta} \ge {\mR}^\theta_{ a^\star_\theta}- q^\star_{\theta} \ge {\mR}^\theta_{ a^\star_{\theta'}}- q^\star_{\theta'} \ge \underline{\mR}^\theta_{\hat a^\star_{\theta'}}-\hat q^\star_{\theta'},
	\]
	where the first and third equality follows by \cref{lm:CB}.
	Similarly,
	\[  	\overline{\mR}^\theta_{\hat a^\star_\theta}-\hat q^\star_{\theta} \ge  {\mR}^\theta_{ a^\star_\theta}- q^\star_{\theta}  \ge0    \]
	Finally, notice that by the IR of agent $\theta$ it must be the case that $\hat q^\star_{\theta}=q^\star_{\theta}\le1$. Hence $\hat q^\star_\theta\in Q_{ a^\star_\theta} \cap [0,1] \subseteq Q^{\delta(1+2/c)}_{ \hat a^\star_\theta}$ by Lemma~\ref{lem:QChangeA}.

	Thus, $V(\hat\PS^\star)\le V(\tilde\PS)$ as $\hat\PS^\star$ is feasible for \Cref{prog:pricingcntdiscrete} while $\tilde\PS$ is its optimum.
	Hence, \[V(\tilde\PS) \ge V(\hat\PS^\star)\ge  V(\PS^\star)-\delta,
	\]
	which concludes the proof.
\end{proof}

%% file: src/appendixRandom.tex
\section{Menus of Randomized Payment Schemes (Section \ref{sec:randomized})}

\propositiondeterministicBad*
\begin{proof}
	Without loss of generality, let $n$ be a given even number, and consider the following instance:\footnote{The construction can be easily extended to an odd $n$ adding a dummy type.}  
	\begin{itemize}[label=$\diamond$]
		\item The set of types is  $\Theta\coloneqq\{\theta_{i,j}\}_{i \in \range{n/2}, j \in \range{2}}$;
		\item The set of actions is $\cA\coloneqq\{a_{i,j}\}_{i \in \range{n/2}, j \in \range{2}}$;
		\item The set of outcomes is $\Omega\coloneqq\{\omega_{i,j}\}_{i \in \range{n/2}, j \in\range{2}}$;
		\item The costs are $c_{a_{i,j}}=0$ for all $a_{i,j}\in\cA$;
		\item The outcome distributions are such that $\mF_{a_{i,j}}$ deterministically induces $\omega_{i,j}$  for all $a_{i,j}\in\cA$;
		\item The rewards are $R_{\theta_{i,j}}(\omega_{k,z})=0$ for each $k\neq i$ and $z\neq j$ and $R_{\theta_{i,j}}(\omega_{k,z})=2^{-i}$ otherwise;
		\item The type distribution $\vxi$ is such $\xi_{\theta_{i,j}}=\frac{2^i}{\sum_{k\in\range{n/2}}2^{k+1}}$. 
	\end{itemize}

	\xhdr{Lower Bound for Randomized Menus.} Now we show that there exists a menu of randomized contracts $\Phi$ with utility at least $ \frac{n}{2}\frac{1}{\sum_{k \in \range{n/2}} 2^{k+1} }$.
	
	For each type $\theta_{i,j}\in\Theta$, we set $\phi^{\theta_{i,j}}_{a_{i,1}}=\phi^{\theta_{i,j}}_{a_{i,2}}=\frac{1}{2}$, and $\vp_{\theta_{i,j},a_{i,1}}$ as the payment scheme with $p_{\theta_{i,j},a_{i,1}}(\omega_{i,1})= 2^{-i-1}$, and payment $0$ on all the other outcomes. Moreover, we set $\vp_{\theta_{i,j},a_{i,2}}$ as the payment scheme with $p_{\theta_{i,j},a_{i,2}}(\omega_{i,2})= 2^{-i-1}$, and payments $0$ on all the other outcomes.
	Each user's type $\theta_{i,j}$ is incentivized to report their type since, by doing so, their expected utility is $2^{-i}-2^{-i -1}=2^{-i-1}$. 
	On the other hand, reporting a type $\theta_{i,k}$, with $k\neq j$, they will get exactly the same randomized contract and thus the same utility. 
	Differently, reporting another type $\theta_{j,k}$ with $j\neq i$, they will get an expected utility of at most $\frac{1}{2} (2^{-i}-2^{-k-1})< 2^{-i-1}$.
	
	It is easy to see that the service provider always plays the action that induces the largest payment, as all the actions have cost $0$ and thus its IC constraints are satisfied.
	Hence, we can conclude that the expected service provider's utility it at least $\sum_{i\in \range{n/2},j \in \range{2}} \frac{2^i}{\sum_{k \in \range{n/2}} 2^{k+1} } 2^{-i-1}=\frac{n}{2}\frac{1}{\sum_{k \in \range{n/2}} 2^{k+1} }$.
	
	\xhdr{Upper Bound for Deterministic Menus.} We conclude the proof showing that any deterministic menu guarantees an service provider's utility at most $8 \frac{1}{\sum_{k \in \range{n/2}} 2^{k+1}}$.
	
	Consider any deterministic menu $\PS=\{(a_\theta, \vp_\theta)\}_{\theta \in \Theta}$.
	For each $j\in \range{2}$, let $\cA^j\coloneqq\{a_{i,j}\}_{i \in \range{n/2}}$. Consider a $j\in\range{2}$, and the set of types $\Theta^j$ that such that $a_{\theta}\in \cA^j$.
	We  prove that all the types in $\Theta^j$ receives the same expected payment $x^j$.
	Indeed, given two types $\theta,\theta' \in \Theta^j$, thanks to the user IC constraints, we have that  
	\[ 	\mF_{a_{\theta}} [\mR_\theta-\vp_\theta]\ge 	\mF_{a_{\theta'}} [\mR_{\theta}-\vp_{\theta'}], \]
	and 
	\[ 	\mF_{a_{\theta'}} [\mR_{\theta'}-\vp_{\theta'}]\ge 	\mF_{a_{\theta}} [\mR_{\theta'}-\vp_{\theta}]. \]
	By the definition of $\Theta^j$, it holds that $\mF_{a_{\theta}} \mR_\theta=\mF_{a_{\theta'}} \mR_\theta$ and $\mF_{a_{\theta'}} \mR_{\theta'}=\mF_{a_{\theta}} \mR_{\theta'}$, and we get that $\mF_{a_{\theta}} \vp_\theta=\mF_{a_{\theta'}} \vp_{\theta'}=x^j$.
	
	Let $\theta_{i,j}$ be the type with largest index $i$ in $\Theta^j$, and call such index $i(j)$. The IR constraint of this type implies that $x^j\le 2^{-i(j)}$.
	Hence, the overall service provider's utility from types in $\Theta^j$ is at most
	\begin{align}
	\sum_{\theta \in \Theta^j} \xi_{\theta}  \mF_{a_{\theta}} [\vp_\theta-c_{a_{\theta}}]     & \le \sum_{i\le i(j)} \sum_{l \in \range{2}} \frac{2^{i}}{\sum_{k \in \range{n/2}} 2^{k+1} } 2^{-i(j)}\notag{}\\
	& \le {4}(2^{i(j)}-1)2^{-i(j)}\frac{1}{\sum_{k \in \range{n/2}} 2^{k+1}}\tag{$\sum_{i=1}^n2^i=2(2^n-1)$}\\
	&\le 4\frac{1}{\sum_{k \in \range{n/2}} 2^{k+1}}.\notag{}
	\end{align}
	Summing over $j \in \range{2}$, we get that the overall service provider's utility is at most $8 \frac{1}{\sum_{k \in \range{n/2}} 2^{k+1}}$, concluding the proof.
\end{proof}

\lemmarandomrelax*
\begin{proof}
	To prove the statement, it is sufficient to show that for each feasible solution $\Phi=\{( \phi^{\theta},(\vp_{\theta,a})_{a \in A})\}_{\theta\in \Theta}$ to Program~\ref{pr:quadratic}, there exists a feasible solution $(\vx,\vphi)$ to LP~\ref{lp:relax} with at least the same value.
	We simply build such solution letting $x_{\theta,a,\omega}= \phi^{\theta}_a  p_{\theta,a}(\omega)$, and keeping the same $\vphi$. It is easy to see that this solution has the same value and satisfies all the constraints. 
\end{proof}

\regularlemma*

\begin{proof}
	Given a $\vphi$ and an $\vx$, let $B(\vphi,\vx)\subseteq \Theta \times \cA$ be the set of $(\theta,a)$ such that $\phi^\theta_a=0$ and $\sum_{\omega\in\Omega}x_{\theta,a,\omega}>0$.
 	Define: 
	\begin{itemize}
		\item $A(\vphi,\vx|\theta)\coloneqq  \{a : (\theta,a)\in B(\vphi,\vx)\} $;
		\item $\Theta(\vphi,\vx|a)\coloneqq\{\theta: (\theta,a)\in B(\vphi,\vx)\}$;
		\item $\Theta(\vphi,\vx)\coloneqq\cup_{a\in\cA}\Theta(\vphi,\vx|a)$. 
	\end{itemize}
	
	First, we show that for each type $\theta \in \Theta(\vphi,\vx)$ it must be the case that $\phi^\theta_\varnothing<1$. Indeed, if  $\phi^\theta_\varnothing<1$, by Constraint~\eqref{eq:relax3} it holds
	\[ x_{\theta,a,\omega}=0  \quad \forall a \in \cA, \omega \in \Omega, \]
	implying $\theta \notin \Theta(\vphi,\vx)$.
	Then, for any $\theta\in\Theta(\vphi,\vx)$, define $\hat a(\theta)$ as an arbitrary action such that $\phi^\theta_{\hat a(\theta)}>0$. As we show above, such an action exists.
	
	Let $(\vx,\vphi)$ be a feasible solution to Program~\ref{lp:relax}. We build a regular solution $(\bar \vx,\bar \vphi)$ as follows. 
	For each $\theta \in \Theta,a \in \cA,\omega\in \Omega$, we set 
	\[
	\bar x_{\theta, a,\omega}=x_{\theta,a,\omega}\  \mathbb{I}\left[a \notin A(\vphi,\vx|\theta)\right]+\sum_{a'\in A(\vphi,\vx|\theta),\omega' \in \Omega}F_{a'}(\omega')x_{\theta,a',\omega'}\  \mathbb{I}[a=\hat a(\theta)].
	\]
	Moreover, we set  $\bar \vphi =\vphi$.
	Straightforward computations let us conclude the following claim:
	\begin{claim}\label{claim:randomreconver}
		Let $(\bar \vx, \bar \vphi)$ be defined as above. Then, the following holds:
		\begin{enumerate}
			\item $\bar x_{\theta,a,\omega}=0$ for all $(\theta, a)\in B(\vphi,\vx)$, and $\omega\in\Omega$;\label{claim:randomreconver1}
			\item $\bar x_{\theta,a,\omega}=x_{\theta, a,\omega}$ for all $\theta\notin \Theta(\vphi,\vx)$, and $a\in\cA$, $\omega\in\Omega$;\label{claim:randomreconver2}
			\item  $\bar x_{\theta,a,\omega}=x_{\theta, a,\omega}$ for all $\theta\in \Theta(\vphi,\vx)$, $a\notin A(\vphi,\vx|\theta)$, and $a\neq\hat a(\theta)$;\label{claim:randomreconver3}
			\item  $\bar x_{\theta,a,\omega}=x_{\theta, a,\omega}+\sum_{a'\in A(\vphi,\vx|\theta),\omega' \in \Omega}F_{a'}(\omega')x_{\theta,a',\omega'}$ for all $\theta\in \Theta(\vphi,\vx)$, $a\notin A(\vphi,\vx|\theta)$, and $a=\hat a(\theta)$.\label{claim:randomreconver4}
		\end{enumerate}
	\end{claim}
	\begin{proof}
		To prove the claim, it is sufficient to apply the definition of $\bar \vx$. 
	\end{proof}
	
	Moreover, we also prove the following claim that states that changing the payments from $\vx$ to $\bar \vx$ does not affect the expected payment for any type $\theta$:
	
	\begin{claim}\label{claim:expectedrandom2}
		For all $\theta\in\Theta$, it holds:
		\[
		\sum\limits_{a\in\cA,\omega\in\Omega} F_a(\omega)\bar x_{\theta,a,\omega}=\sum\limits_{a\in\cA,\omega\in\Omega} F_a(\omega)x_{\theta,a,\omega}.
		\]
	\end{claim}
	
	\begin{proof}
		Consider the following chain of equalities:
		\begin{align*}
			\sum\limits_{a\in\cA,\omega\in\Omega}& F_a(\omega)\bar x_{\theta,a,\omega}\\
			&=\sum\limits_{a\in\cA,\omega\in\Omega} F_a(\omega)\left[x_{\theta,a,\omega} \mathbb{I}\left[a \notin A(\vphi,\vx|\theta)\right]+\sum_{a'\in A(\vphi,\vx|\theta),\omega' \in \Omega}F_{a'}(\omega')x_{\theta,a',\omega'} \mathbb{I}[a=\hat a(\theta)]\right]\\
			&=\sum\limits_{a\in\cA,\omega\in\Omega} F_a(\omega)x_{\theta,a,\omega} \mathbb{I}\left[a \notin A(\vphi,\vx|\theta)\right]+\sum_{\omega\in\Omega}F_{\hat a(\theta)}(\omega)\sum_{a'\in A(\vphi,\vx|\theta),\omega' \in \Omega}F_{a'}(\omega')x_{\theta,a',\omega'}\\
			&=\sum\limits_{a\not\in A(\vphi,\vx|\theta),\omega\in\Omega} F_a(\omega)x_{\theta,a,\omega} +\sum_{a'\in A(\vphi,\vx|\theta),\omega' \in \Omega}F_{a'}(\omega')x_{\theta,a',\omega'}\\
			&=\sum\limits_{a\in \cA,\omega\in\Omega} F_a(\omega)x_{\theta,a,\omega}
		\end{align*}
		concluding the proof of the claim.
	\end{proof}

	Then, we show that $(\bar \vphi,\bar \vx)$  is indeed regular. Take any $(\theta,a)\in B(\vphi,\vx)$. We have that $\bar \phi^{\theta}_a= \phi^{\theta}_a=0$, and \Cref{claim:randomreconver}-(\ref{claim:randomreconver1}) assures that $\bar x_{\theta,a,\omega}=0$ for all $\omega \in \Omega$. Thus, $(\bar \vx, \vphi)$ is regular.
	
	Finally, we are left to show that $(\bar \vx,\bar \vphi)$ satisfies Constraints~\eqref{eq:relax2} to \eqref{eq:relax3}.
	
	\xhdr{Constraints~\eqref{eq:relax2}.}
	Consider a $\theta\in \Theta(\vphi,\vx)$, an $a\notin A(\vphi,\vx|\theta)$ such that $a=\hat a(\theta)$, and an $a''\in \cA$. Here,  \Cref{claim:randomreconver}-\eqref{claim:randomreconver4} applies. Thus:
	\[
	\bar x_{\theta,a,\omega}=x_{\theta, a,\omega}+\sum_{a'\in A(\vphi,\vx|\theta),\omega' \in \Omega}F_{a'}(\omega')x_{\theta,a',\omega'}
	\]
	
	Now, consider the following inequalities:
	\begin{align}
		\sum_{\omega\in\Omega}& F_{a}(\omega) (\bar x_{\theta,a,\omega}-\phi_a^\theta c_a)=\sum_{\omega\in\Omega} F_{a}(\omega) \left(x_{\theta, a,\omega}+		\sum_{a'\in A(\vphi,\vx|\theta),\omega' \in \Omega}F_{a'}(\omega')x_{\theta,a',\omega'}-\phi_a^\theta c_a\right)\notag{}\\
		\ge& \sum_{\omega\in\Omega} F_{a''}(\omega) (x_{\theta, a,\omega} -\phi_a^\theta c_{a''})+ \sum_{\omega\in\Omega} F_{a}(\omega) \left(\sum_{a'\in A(\vphi,\vx|\theta),\omega' \in \Omega}F_{a'}(\omega')x_{\theta,a',\omega'}\right)\tag{$\vx$ is feasable}\\
		=&\sum_{\omega\in\Omega} F_{a''}(\omega) (x_{\theta, a,\omega} -\phi_a^\theta c_{a''}) + \sum_{a'\in A(\vphi,\vx|\theta),\omega' \in \Omega}F_{a'}(\omega')x_{\theta,a',\omega'}\tag{$\mF_{a}\in\Delta(\cA)$}\\
		=&\sum_{\omega\in\Omega} F_{a''}(\omega) (x_{\theta, a,\omega} -\phi_a^\theta c_{a''})+\left(\sum_{\omega\in\Omega} F_{a''}(\omega)\right) \left(\sum_{a'\in A(\vphi,\vx|\theta),\omega' \in \Omega}F_{a'}(\omega')x_{\theta,a',\omega'}\right)\tag{$\mF_{a''}\in\Delta(\cA)$}\\
		=&\sum_{\omega\in\Omega} F_{a''}(\omega) \left(x_{\theta, a,\omega}+
		\sum_{a'\in A(\vphi,\vx|\theta),\omega' \in \Omega}F_{a'}(\omega')x_{\theta,a',\omega'}-\phi_a^\theta c_{a''}\right)\notag{}\\
		=&\sum_{\omega\in\Omega} F_{a''}(\omega) (\bar x_{\theta, a,\omega}-\phi_a^\theta c_{a''}).\notag{}
	\end{align}
	
	Thus, $\sum_{\omega\in\Omega} F_{a}(\omega) \bar x_{\theta,a,\omega}\ge \sum_{\omega\in\Omega} F_{a''}(\omega) \bar x_{\theta,a,\omega}$.
	
	Consider any other tuple $(\theta,a,a'')$ not yet considered. In this case, we are in one of the cases \eqref{claim:randomreconver1} to \eqref{claim:randomreconver3} of \cref{claim:randomreconver}.
	Hence, we either have $\bar x_{\theta,a,\omega}=0$ (and $\bar \phi^\theta_{a}=0$), or $x_{\theta,a,\omega}=\bar x_{\theta,a,\omega}$. It is easy to see that this implies that the constraint is satisfied.
	
	Hence, we show that, for all $a,a''\in\cA$ and $\theta\in\Theta$, Constraints~\eqref{eq:relax2} is satisfied by $(\bar\vx, \bar \vphi)$.

	\xhdr{Constraints~\eqref{eq:relax1}.}
	By using \Cref{claim:expectedrandom2} twice, and $\bar\vphi=\vphi$, we can easily conclude that for each $\theta,\theta'\in \Theta$:
	\begin{align}
		\sum_{\omega\in\Omega}\sum_{a\in\cA}  F_{a}(\omega) (\bar \phi^{\theta}_a R_\theta(\omega)- \bar x_{\theta,a,\omega})&=\sum_{\omega\in\Omega}\sum_{a\in\cA}  F_{a}(\omega) ( \phi^{\theta}_a R_\theta(\omega)- x_{\theta,a,\omega})\tag{by \Cref{claim:expectedrandom2}}\\
		&\ge \sum_{\omega\in\Omega}\sum_{a\in\cA}  F_{a}(\omega) ( \phi^{\theta'}_a R_\theta(\omega)- x_{\theta',a,\omega})\tag{$(\vphi,\vx)$ is feasible}\\
		&\ge \sum_{\omega\in\Omega}\sum_{a\in\cA}  F_{a}(\omega) ( \bar \phi^{\theta'}_a R_\theta(\omega)- \bar x_{\theta',a,\omega}).\tag{by \Cref{claim:expectedrandom2}},
	\end{align}
	
	\xhdr{Constraints~\eqref{eq:relax3}.}
	By using \Cref{claim:expectedrandom2} we can easily conclude the following, which holds for all $\theta\in\Theta$:
	\begin{align}
		\sum_{\omega\in\Omega}\sum_{a\in\cA}  F_{a}(\omega) (\bar \phi^{\theta}_a R_\theta(\omega)- \bar x_{\theta,a,\omega})&=\sum_{\omega\in\Omega}\sum_{a\in\cA}  F_{a}(\omega) ( \phi^{\theta}_a R_\theta(\omega)- x_{\theta,a,\omega})\tag{by \Cref{claim:expectedrandom2}}\\
		&\ge 0\tag{$(\vphi,\vx)$ is feasible},
	\end{align}
	where we also used that $\bar\vphi=\vphi$.
	
	\xhdr{Objective function.}
	
	Similarly to the previous constraints, we can easily prove that the objective function of $(\bar\phi,\bar \vx)$ coincides with the one of $(\vphi,\vx)$ Formally,
	
	\begin{align*}
		\sum_{\theta\in\Theta}\sum_{a\in\cA}\sum_{\omega\in\Omega}\xi_\theta F_a(\omega)\left(\bar x_{\theta,a,\omega}-\bar \phi_a^\theta c_{a}\right)&=\sum_{\theta\in\Theta}\sum_{a\in\cA}\sum_{\omega\in\Omega}\xi_\theta F_a(\omega)\left(x_{\theta,a,\omega}-\bar \phi_a^\theta c_{a}\right)\tag{by \Cref{claim:expectedrandom2}}\\
		&=\sum_{\theta\in\Theta}\sum_{a\in\cA}\sum_{\omega\in\Omega}\xi_\theta F_a(\omega)\left(x_{\theta,a,\omega}- \phi_a^\theta c_{a}\right)\tag{as $\vphi=\bar\vphi$}
	\end{align*}
	
	\xhdr{Wrapping Up.} We proved that $(\bar\vx,\bar\vphi)$ is regular, it has the same objective of $(\vx,\vphi)$, and it is feasible. This concludes the proof.
\end{proof}

\regularLemmaDue*

\begin{proof}
	Given a regular solution $(\bar \vphi,\bar \vx)$ to Program~\ref{lp:relax}, consider the following solution $(\vphi,\vp)$ to Program~\ref{pr:quadratic}.
	For each $\theta \in \Theta,a \in \cA$, and $\omega\in \Omega$, we set $p_{\theta,a,\omega}= \bar x_{\theta,a,\omega}/\bar \phi^\theta_a$ if $\bar \phi^\theta_a\neq0$, and $p_{\theta,a,\omega}=0$ otherwise. Moreover, we set $\vphi=\bar \vphi$.
	
	Since $\bar \vphi,\bar \vx$ is regular, by Definition~\ref{def:regular} we have that $\bar \phi^\theta_a\neq 0$ whenever there exists an $\omega \in \Omega$ such that $x_{\theta,a,\omega}\neq0$. Hence, $\phi^{\theta}_a  p^{\theta,a}_\omega= x_{\theta,a,\omega}$ for each $\theta \in \Theta,a \in \cA$, and $\omega\in \Omega$. It is easy to check that this implies that $(\vphi,\vp)$ is a feasible solution to Program~\ref{pr:quadratic} with value at least $\sum_{\theta} \xi_\theta \sum_{\omega,a} F_{a}(\omega) (\bar x_{\theta,a,\omega}- \bar \phi^{\theta}_a c_{a})$.
\end{proof}